\documentclass[10pt]{article}

\usepackage{amsfonts}

\usepackage{amsmath}

\usepackage{amssymb}

\usepackage{amstext}

\usepackage{amsthm}

\usepackage{bbm}

\usepackage{bbold}

\usepackage{bussproofs}

\usepackage{comment}

\usepackage{enumitem}

\usepackage{epigraph}

\usepackage[colorlinks=true,
	citecolor=blue,
	urlcolor=blue]{hyperref}
\DeclareUrlCommand\url{\color{blue}}

\usepackage[a4paper, total={170mm,257mm},
 left=20mm,
 top=20mm,]{geometry}

\usepackage{hhline} 

\usepackage{mathbbol} 

\usepackage{mathtools} 

\usepackage{cmll} 

\usepackage{mathrsfs} 

\usepackage{pgfplots}
\pgfplotsset{compat=newest}

\usepackage{ragged2e}
\justifying

\usepackage{stmaryrd} 

\usepackage{tikz}

\usepackage{tikz-cd}

\usepackage{quiver}

\usepackage{float}

\usepackage{adjustbox}

\usepackage[normalem]{ulem} 

\usepackage{xcolor}

\usepackage[all]{xy}
\xyoption{2cell}
\xyoption{curve}
\UseAllTwocells
\SelectTips{cm}{}

\theoremstyle{definition}
\newtheorem{theorem}{Theorem}[section]

\newtheorem{corollary}[theorem]{Corollary}
\newtheorem{definition}[theorem]{Definition}
\newtheorem{example}[theorem]{Example}
\newtheorem{lemma}[theorem]{Lemma}
\newtheorem{proposition}[theorem]{Proposition}
\newtheorem{remark}[theorem]{Remark}

\def\kat#1{{\mathscr{#1}}}

\def\A{\kat{A}}
\def\B{\kat{B}}

\def\M{\kat{M}}
\def\N{\kat{N}}
\newcommand{\V}{\mathcal{V}}

\def\mon{\mathtt{M}}

\def\act{\mathbf{Act}}
\def\mact{\M\hspace{-1.5pt}\mbox{-}\hspace{.5pt}{\mathbf{Act}}}

\newcommand{\set}{\mathbf{Set}}
\newcommand{\cat}{\mathbf{Cat}}

\newcommand{\id}{\mathsf{id}}

\newcommand{\op}{^{\mathsf{op}}}

\newcommand{\st}{\mathsf{st}}
\newcommand{\cst}{\mathsf{cst}}

\newcommand{\optic}{\mathbf{Optic}}

\newcommand{\one}{\mathbb 1}
\newcommand{\two}{\mathbb 2}

\renewcommand{\to}{\longrightarrow}
\newcommand{\nto}{\nrightarrow}

\begin{document}

%%===============================%

\title{On Graded Monads, Distributive Laws and Costrong Functors}

%%===============================%
%\author{}

\author{Adriana Balan\thanks{Support by the research project GNaC2023 ARUT 88/11.10.2023 {\em Tensor Products and Nuclearity} is \linebreak acknowledged.} \thanks{Email: {\tt adriana.balan@upb.ro}. Department of Mathematical Methods and Models 
\& Fundamental Sciences Applied in Engineering Research Center, National University of Science and Technology POLITEHNICA Bucharest, Romania.}
\qquad
Silviu-George Pantelimon\thanks{Email: silviu.pantelimon@upb.ro. Department of Computer Science, National University of Science and \linebreak Technology POLITEHNICA Bucharest, Romania.}}

\maketitle

%%===============================%

{\small{{\bf Abstract.} 
Strong functors and monads are ubiquitous in Computer Science. 
More recently, (strong) comonads have demonstrated their use in structuring context-dependent notions of computation.
However, the dualisation of ``being strong'' property passed somehow unobserved so far. 
We argue that ``being costrong'' gives a different understanding of how functors can interact with monoidal structures.
We shall see that the well-known correspondence between distributive laws $FT \to TF$ of an endofunctor $F$ over a monad $T$, on one hand, and extensions of $F$ to the Kleisli category of that monad, on the other hand, generalises from ordinary monads to graded ones. 
The gist here is to recognise that the costrength of a costrong functor is nothing but a ``graded'' distributive law. 
As such, ``being costrong'' is a structure that a functor may have. 
Examples of costrong functors with respect to different graded monads are provided, with emphasis to the cartesian case, and applications to optics and coalgebras are given. 
}}

%%===============================%

{\small{{\bf Keywords}. Graded monads, distributive laws, (co)strong functors, actions of monoidal categories, optics 

}}

%%===============================%

\section{Introduction}

\paragraph{Graded monads and (co)strong functors.} 
A monoid $M$ in a monoidal category $(\M,\otimes,I)$ induces a {\tt Writer}-like monad $T_M = M \otimes -$ on $\M$ expressing how effects of type $M$ may accumulate during computations. 
Letting $M$ run through all objects of $\M$ produces a family of endofunctors $T_M$ on $\M$, not necessarily all of them monads, but which, in their totality, still behave like one. 
That is, there are natural transformations $\eta: \id_\M \to T_I$, $\mu:T_M  T_N\to T_{M\otimes N}$ satisfying associativity and unit-like relations. 
Such a structure is called an $\M$-graded monad~\cite{FujiiKatsumataMelies-FoSSaCS2016,Katsumata2014}. 
Effects of {\em different} types $M$ and $N$ can now be combined using the monoidal structure into an effect of type $M \otimes N$. 

\medskip

A strong functor $F$ on the monoidal category $\M$ comes equipped with a family of arrows $\st_{M,X}:M \otimes FX \to F(M \otimes X)$ ({\em strength}) natural both in $M$ and in $X$, that pushes computational effects (of type $M$) inside the functor, seen as a datatype constructor. 
The dual concept -- though far less prominent -- is that of a {\em costrong} functor. It possesses a natural transformation, called {\em costrength} and  pointing in the opposite direction $\cst_{M,X}:F(M \otimes X)\to M \otimes FX$ that allows extraction of effects of type $M$, thus providing direct access to them~\cite{BluteCockettSeely1997}. 

\medskip

Expressing the natural transformations $\st_{M,-}$ or $\cst_{M,-}$ as $ \st_{M,-}:T_M  F \to F  T_M$, respectively $ \cst_{M,-}:F  T_M \to T_M  F$, 
where $T_M$ is the {\tt Writer}-like functor mentioned earlier, 
allows us to notice that a (co)strength of a (co)strong functor is precisely the {\em graded} analogue of the usual {\em distributive law} between a monad and a functor (or viceversa). 
In the denotational semantics of programming languages, distributive laws relating a monad and an endofunctor are the same as liftings to the Eilenberg-Moore category of algebras for the monad, respectively to the Kleisli category of the monad, depending on the orientation of the distributive law. 
These liftings provide modular ways to reason about composed effects: Eilenberg-Moore liftings emphasise algebraic models (for example, by resolving effects into structured algebras), while Kleisli liftings feature computational arrows (such as sequencing of effects in programming). 
A particular successful use of the latter was the monadic fold. The key ideas were to view a (recursive) datatype as the initial algebra of a functor $F$, and to lift it to the Kleisli category of a monad $T$ using a distributive law $FT \to TF$, obtaining thus (monadic) catamorphisms that act on $T$-valued arrows~\cite{Fokkinga1994}. 
Another important application of distributive laws of type $FT \to TF$ arose from coalgebraic trace semantics~\cite{HasuoJacobsSokolovaCMCS2006}. 
Transition systems are usually modelled as $F$-coalgebras, and their traces (sequences of observable actions or behaviours) are obtained by mapping into the final coalgebra\footnote{Final coalgebra of the lifted functor.} in the Kleisli category of a suitable monad $T$. 
This monad naturally encodes 'branching' effects (like the powerset monad for nondeterminism, or the distribution monad for probability), enabling a uniform definition of traces for various transition systems.

%======================================%

\paragraph{Costrength: overlooked despite duality.}
The asymmetric study of strong and costrong functors reflects a computational reality: every functor on the cartesian category of sets and functions  admits a (unique) strength, which naturally models the forward flow of data through processing pipelines, a fundamental operation in programming.
In other words, strength merges data into a computational process from an outside source — the same way lax monoidal functors combine two in-parallel functor applications, or monad multiplication merges nested data — thereby providing more possibilities for processing data.
By contrast, costrength extracts data from a functor-abstracted process, allowing observation of its internals — such as logging events — or recovery of intermediate results for downstream processing. 
The resulting duality slogan is: strength pushes data {\em into} computation, while costrength pulls data {\em out of} it. 
Despite this duality, the relative neglect of costrength has a structural explanation. 
Strong functors are intimately connected to enrichment: a functor on a closed monoidal category is strong if and only if it is enriched over that category, and in functional programming this corresponds to the familiar ability to lift a function $X \to Y$ to $FX \to FY$ using the internal hom. 
The dual notions of cohom, and costrong functors -- that is, functors with the ability of extracting a (co)function out of functor application -- rarely appear in functional programming, and so costrength has lacked an equally natural computational counterpart. 
Moreover, in the cartesian case, we shall see that costrong functors coincide precisely with copointed ones: functors $F$ equipped with a natural transformation $FX \rightarrow X$. 
This might suggest that costrength is a weak or uninteresting structure, but this interpretation is mistaken. 
Copointed functors are exactly those that allow observation and extraction of data out of a functor-abstracted computation — loggers, readers, and observable transition systems are paradigmatic examples. 
It is only that, in the Cartesian setting, this structure is so common that it is not explicitly named.
Currently, however, more and more monoidal structures are employed in the semantics of programming languages, and it is here that {\em hides the strength} of costrong functors.

%======================================%

\paragraph{Costrength nowadays.}
The current moment is opportune for several reasons.
The graded (co)monad semantics of algebraic effects and (co)effect handlers has matured to the point where the fine structure of distributive laws is practically relevant, not merely abstractly interesting. 
Concurrently, the theory of optics (lenses, prisms, and their generalisations) has drawn attention to (co)strength-like insertion/extraction maps as categorical transformers of bidirectional data access. 

%======================================%

\paragraph{Costrength perspective.}

There are two main aspects to this. 
On the theoretical side, the graded generalisation reveals that strength and costrength are not merely dual curiosities but organise into a coherent 2-categorical structure governing how effects interact with datatypes — a structure that the ungraded theory obscures. 
On the practical side, costrong functors provide the right abstraction for effects that expose their internal state: loggers, profilers, observable transition systems, and read-only environments are all, in categorical terms, costrong endofunctors, and the theory developed here gives a uniform framework for reasoning about their composition and lifting.

%======================================%

\paragraph{Outline.}
In this paper, we pursue two complementary perspectives on costrong functors that emerge from the preceding discussion: %
first, graded distributive laws of endofunctors over graded monads and their correspondence with liftings to Kleisli categories; %
second, the dualisation of the well-known theory of strong functors familiar from functional programming. 
These two viewpoints are, in essence, two sides of the same coin--or, to borrow a well-known metaphor, two sides of the same elephant--and we believe that much more remains to be said on the subject.

\medskip

The paper is structured as follows: %
After a very brief review of monoidal categories, in Section~\ref{sec:gr-mnd} we discuss graded monads, (graded) distributive laws of endofunctors over graded monads, and the one-to-one correspondence with liftings to the associated Kleisli category of the graded monad, as introduced in~\cite{FujiiKatsumataMelies-FoSSaCS2016}. 
In the subsequent Section~\ref{sec:M-acts} we adopt a different terminology: 
Rather than speaking of graded monads and functors endowed with (graded) distributive laws, we speak of (lax) actions of monoidal categories and costrong functors as morphisms between actions of monoidal categories. 
Subsection~\ref{sec:costrong=copointed} is about costrong endofunctors on cartesian categories acting upon themselves. We show that costrengths are in one-to-one correspondence with copoints, that is, natural transformations to the identity functor. 
In particular, all comonads and all copointed endofunctors on cartesian categories are costrong, possibly in more than one way. 
Next, we explore categorical aspects of costrong (endo)functors in Subsection~\ref{sec:acts = algs}. 
Readers primarily interested in applications may wish to proceed directly to Section~\ref{sec:applications}, of which Subsection~\ref{sec:optics} is the authors' initial motivation for this paper and it concerns morphisms of (mixed) optics~\cite{CEGLMPR2020}. 
Originating in functional programming, optics have emerged as a powerful formalism for bidirectional transformations, capturing in a compositional manner how data is be accessed, modified, and propagated through complex structures.
Transforming optics along a pair of functors, one strong and the other costrong, preserves the underlying bidirectional behaviour while allowing the surrounding computational context to vary. 
In particular, it provides a conceptual reason for studying the notion of costrength. 
The second application (Subsection~\ref{sec:costrong-stream}) explores the interaction between costrong functors and stream-based computations.
A $(\set,\times,1)$-costrong functor $F$ yields a distributive law over the graded {\tt Writer} monad $T_M(X)=M\times X$, allowing $F$ to lift stream coalgebras while preserving their observable output behaviour: 
applying $F$ wraps the state space in a context, but the stream of outputs remains accessible via the costrength and is interpreted through the final stream coalgebra $M^\omega$.
Moreover, the induced $F$-algebra structure 
on the $M^\omega$ provides the basis for a coinduction up-to principle, modular in the parameter $M$: 
every $M\times F(-)$-coalgebra admits a unique behavioural map into $M^\omega$, enabling coinductive definitions and proofs modulo the context described by the functor $F$.

\medskip

The present paper is a revised and significantly extended version of~\cite{BalanPantelimon2025}. In particular, Section~\ref{sec:gr-mnd}, devoted to graded monads and distributive laws, is entirely new and did not appear in the version presented at FROM~2025.

%======================================%

\section{Monoidal categories}
\label{sec:mon-cats-actions}

We shall assume some basic knowledge of category theory, otherwise we refer the reader to~\cite{MacLane1998}. Throughout the paper, $(\M,\otimes, I)$ will denote a monoidal category. 
For simplicity, we shall not label the associativity and unit constraints $(M\otimes N) \otimes P \cong M \otimes (N \otimes P)$ and $I \otimes M \cong M \cong M \otimes I$ or even omit them when the context is clear, writing as if the tensor product were strict. Monoidal functors whose structural morphisms are isomorphisms will be called pseudo-monoidal, rather than strong monoidal as usual, to distinguish them from the strong functors. 

\medskip

To avoid cluttering the notation, functor composition will be simply written by juxtaposition, using the $\circ$ symbol only when necessary. Likewise, we shall omit parentheses in functor application, denoting $FX$ instead of $F(X)$, except where required for clarity.

\begin{example}
	\begin{enumerate}
	
		\item Given a category $\A$, the category $[\A,\A]$ of endofunctors on $\A$ is strict monoidal with functor composition as tensor product, and unit the identity functor.
		\footnote{We shall ignore size issues that might prevent $[\A,\A]$ from being a legitimate category, restricting to accessible endofunctors if necessary.}

		\item Any cartesian category is symmetric monoidal with binary product $X \times Y$ as tensor, and unit the terminal object $\one$. 
        Dually, any cocartesian category is monoidal with respect to the binary coproduct $X+Y$ and initial object $\mathbb 0$.

		\item Any monoid  $(\mon,*,e)$ is a discrete monoidal category: the elements of $\mon$ are the objects of this category, there are no morphisms except the identities, with monoidal structure given by the binary operation $*$ of the monoid. The unit $e$ is the  monoidal unit. 
		
		If additionally the monoid is ordered (in particular, the multiplication is monotone in both arguments), then the monoidal category described above is no longer discrete, the morphisms being provided by the order relation. %

\end{enumerate}
\end{example}

%======================================%

\section{Graded Monads and Distributive Laws}\label{sec:gr-mnd}

It goes back to Moggi's seminal work that computational effects, like nondeterminism, storage or exceptions can be uniformly captured using monads~\cite{Moggi1991}. Intuitively, a monad $T$ associates to each type $X$ a type $TX$ of computations of type $X$; a function with side effects which takes inputs of type $X$ and returns values of type $Y$ is just a function of type $X\to TY$. 
It became clear in the last years that the monad  approach can be refined to track quantitative information about the effects performed by a computation (see e.g.~\cite{Katsumata2014,WadlerThiemann2003}). Mathematically, this is achieved by annotating ({\em grading}) the monadic types by effects $M$ in the form $T_M X$, and a monoidal structure on effects is necessary to ensure that these can be composed. For example, there is a {\em graded state} monad, where the grade tracks how many times a reference is used. 
There is also a {\em graded} version of the {\em list} monad, by considering lists of length less than $n$, for each natural number $n$.

\begin{definition}
Let $(\M,\otimes,I)$ be a monoidal category and $\A$ be a category. 
An $\M$-graded monad on $\A$~\cite{FujiiKatsumataMelies-FoSSaCS2016}, also called {\em parametrised monad} in~\cite{Katsumata2014,Mellies2012}, is a lax monoidal functor 
\[
T: (\M,\otimes, I) \to ([\A,\A],\circ,\id_\A) 
\]
Explicitly, an $\M$-graded monad  consists of the following:
\begin{itemize}
\item For each object $M$ in $\M$, an endofunctor $T_M:\A \to \A$, and for each arrow $f:M\to N$ in $\M$, a natural transformation $T_f: T_M \to T_N$, such that the assignment $M \longmapsto T_M$ is functorial;
\item A natural transformation $\eta:\id_\A \to T_I$, called the {\em unit} of the graded monad;
\item For each pair of objects $M,N$ of $\M$, a natural transformation $\mu_{M,N}:T_M  T_N \to T_{M\otimes N}$, called the {\em multiplication} of the graded monad, which is moreover natural in $M$ and $N$;
\end{itemize}
such that the diagrams below commute for every objects $M,N,P$ of $\M$:

\begin{equation}
\begin{gathered}
\xymatrix{T_M  T_N  T_P \ar[r]^{\mu_{T_P}} \ar[d]_{T_M \mu} & T_{M \otimes N} T_P \ar[d]^\mu 
\\
T_M T_{N \otimes P}  \ar[r]^\mu & T_{M \otimes N\otimes P}
} 
\\
\xymatrix{\id_\A T_M \ar[r]^-{\eta_{T_M}} \ar@{=}[dr] & T_I  T_M \ar[d]^\mu
\\
& T_M }
\quad \xymatrix{T_M \id_\A \ar[r]^{T_M \eta} \ar@{=}[dr] & T_M  T_I \ar[d]^\mu
\\
& T_M }
\end{gathered}
\label{eq:graded-mnd}
\end{equation}
\end{definition}
If $\eta$ and $\mu$ are isomorphisms, we recover the notion of an action of the monoidal category $\M$, also known as an $\M$-actegory~\cite{CapucciGavranovic2022,McCrudden2000,Pareigis1977}, that we shall consider in Section~\ref{sec:M-acts}.
If moreover $\eta$ and $\mu$ are identities, we shall call the action (or the graded monad) {\em strict.}

\medskip

The name {\em graded monad} is justified by the fact that if $\M$ is the category $\one$ with one object and one arrow, then an $\M$-graded monad $T$ on a category $\A$ is just a monad on $\A$~\cite{Benabou1967}. 

\begin{remark} 
In fact, taking the grade to be the unit $I$ of the monoidal category $\M$, the $I$-th component of a graded monad $T:\M \to [\A,\A]$ is always a monad on $\A$.  
More generally, $T_M$ is a monad whenever $M$ is a monoid in $\M$.  
\end{remark}

\begin{example}\label{ex:gr-mnds}
\begin{enumerate}
\item (Graded Maybe~\cite{solberg1994})\label{ex:gr-maybe}
The {\tt Maybe} monad is often used in functional programming to encapsulate computations that might fail. 
It allows a graded version, built as follows: 
Let 
\[
\mon = \{\mathtt{f}(ailure), \mathtt{s}(uccess),\mathtt{m}(aybe)\} 
\] 
be a three-element set with partial order \( \mathtt{f} \leq \mathtt{m} \), \( \mathtt{s} \leq \mathtt{m} \), and binary operation 
\[
\begin{array}{c|ccc}
\otimes & \mathtt f & \mathtt s & \mathtt m
\\
\hline
\mathtt f & \mathtt f & \mathtt f & \mathtt f
\\
\mathtt s & \mathtt f & \mathtt s & \mathtt m
\\
\mathtt m & \mathtt f & \mathtt m & \mathtt m
\end{array}
\]
It is easy to check that $((\mon,\le),\otimes,\mathtt s)$ is an  ordered monoid. 
Then a graded monad ${\tt GMaybe}: \mon \to [\set,\set]$ can be defined as follows: ${\tt GMaybe}_{\mathtt f} X = \one$, ${\tt GMaybe}_{\mathtt s} X = X$ and ${\tt GMaybe}_{\mathtt m} X = \mathtt{Maybe} \, X = \one + X$, where $\one=\{*\}$ denotes a singleton set. The inequalities $\mathtt{f} \leq \mathtt{m}$ and $\mathtt{s} \leq \mathtt{m}$ are mapped by $\mathtt{GMaybe}$ to the canonical injections $\one \to \one +X$, respectively $X \to \one +X$.

\item (Graded Exceptions (I)) The  most common example of a monadic computation involves exception monads. In a programming language we would typically define a type for the coproduct $T_E X = E + X$, where $E$ represents possible exception and $X$ is the type of the successful result. A monadic computation built on this type would execute a program using $\mathtt{map}$ and $\mathtt{bind}$ (or $\mathtt{join}$) and stop when an exception is encountered. 
A difficulty arises, however, when we need to handle exceptions belonging to different types. A standard monad is insufficient in this case, because its exception type is fixed and cannot naturally combine or unify multiple distinct exception types. The graded exception $(\set,+,\emptyset)$-monad on $\set$, $T_E X = E + X$, solves this issue by taking into account all possible exception types and combining them to a disjoint union. 

\item (Graded Exceptions (II)~\cite{Katsumata2014}) In many situations, it is the case that the (large) class of possible exceptions of the previous example can be limited to a known (non-empty) set $E$. Then  the previous example can be refined as follows: Consider the ordered monoid $((\mathcal P(E), \subseteq), \cup,\emptyset)$ of subsets of $E$, with union as binary operation. This induces a graded monad on $\set$ by $T_e X = e+X$, for a subset $e\subseteq E$ of exceptions. The unit $\eta_X: X \cong \emptyset+X$ embeds pure values as successful computations. The multiplication $\mu_{e,e'}:(T_e T_{e'} )X = e+e'+X \to T_{e\cup e'}X = (e\cup e')+X$ propagates (sets of) exceptions by taking their union, thus tracking which exceptions may be raised, but not how many times these occur. In the previous example, the monoidal product was given by disjoint union. Therefore each exception carried its own log, recording the stage of the computation at which was raised.

\item (Graded Writer (I)) Consider $(\set, \times, \one)$ with its cartesian monoidal structure and the graded {\tt Writer} monad $T_MX = M \times X$ described in the Introduction (for an arbitrary monoidal category $(\M,\otimes, I)$); its multiplication and unit are induced by the associativity and unit constraints of the cartesian product. Consequently, the outputs $M_1, \ldots, M_n$ (which can represent logs, traces, warnings, etc.) accumulate to tuples taking values in $M_1 \times \ldots \times M_n$.

\item (Graded Writer (II)~\cite{Katsumata2014}) As for graded exceptions earlier, the graded {\tt Writer} monad admits a variant that is both more restricted and more refined, as follows: Let $\Sigma$ be a set (an alphabet), and consider the ordered monoid of languages  $((\mathcal P(\Sigma^*), \subseteq), \cdot,\{\epsilon\})$, with binary operation induced by string concatenation:
\[
L \cdot L' = \{\omega \omega'\mid \omega \in L, \omega' \in L'\}
\]
The unit is the singleton language consisting of the empty string $\epsilon$. A $\mathcal P(\Sigma^*)$-graded monad $T$ on $\set$ is obtained as follows: for each language $L\subseteq \Sigma^*$, put  $T_L X = L \times X$. The unit and multiplication are given by $\eta_X: X \cong T_{\{\epsilon \}}X = \{\epsilon\} \times X\ , \  \eta_X(x)=(\{\epsilon\},x)$, respectively
\begin{align*}
& \mu_{L,L';X}: T_L T_{L'}X =L\times (L'\times X) \to T_{L \cdot L'}X = (L\cdot L')\times X 
\\
& \mu_{L,L';X}(\omega, (\omega',x)) = (\omega \omega',x)
\end{align*}

\item (Graded List~\cite{McDermottUustalu2022a}) The $((\mathbb N,\le)\cdot, 1)$-graded list monad on $\set$ maps each grade (natural number) $n$ and each set $X$ to the set $\mathsf{List}_{\le n}X$ of lists over $X$ of length at most $n$; for $m\le n$, the natural transformation $\mathsf{List}_{\le m}X\to \mathsf{List}_{\le n}X$ is the inclusion. The unit of the graded monad is again the inclusion $X\to \mathsf{List}_{\le 1}X$, while multiplication is list concatenation.

\item \label{ex-rel-presh}
Let $((\V,\le),\otimes,e)$ be a commutative quantale (a commutative monoid in the category of suplattices). Following Lawvere~\cite{Lawvere1973}, we shall perceive (small) $\V$-enriched categories as generalised metric spaces with $\V$-valued distances, and $\V$-functors as non-expanding maps. Such a generalised metric space will be written $(X,d_X)$ where the $\V$-valued distance $d_X$ is transitive and reflexive, but not necessarily symmetric; additionally, points at zero-distance (where  zero stands for the bottom element of the quantale) are not necessarily equal. Let $\mathbf{Met}(\V)$ denote the category of (small) $\V$-enriched categories and $\V$-functors. 
A $\V\op$-graded monad on $\mathbf{Met}(\V)$ can be obtained as follows: let $v\in \V$ and $(X,d_X)$ a $\V$-metric space and put $T_r X = \{(x_1,x_2)\in X \times X\mid d_X(x_1,x_2)\ge r\}$, with infimum $\V$-metric\footnote{The infimum being computed in $\V$, not in $\V\op$.}
\[
d_{T_r X}((x_1,x_2),(y_1,y_2))= d_X(x_1,y_1)\wedge d_X(x_2,y_2)
\]
For $r \le s$, $T_s X$ is a $\V$-metric subspace of $T_r X$ (hence the $\mathsf{op}$ in $\V$). The unit of this graded monad is the diagonal $\eta_X :X \to T_e X$, $\eta_X(x)=(x,x)$, while the multiplication is given by 
\[
\mu_{r,s;X}:(T_r T_s)X \to T_{r \otimes s}X \ , \ \mu_{r,s}((x_1,x_2),(y_1,y_2)) = (x_1,y_2)
\]
It is not difficult to see that both $\eta_X$ and $\mu_{r,s;X}$ are $\V$-functors and natural in $X$, respectively in $r,s\in \V$ and $X$, and that they indeed induce a $\V\op$-graded monad structure on $T$.

Observe that each $T_rX$ is a binary relation on $X$, and that $T_s X\subseteq T_r X$ for $r\le s$ accounts for dependencies between these binary relations (like causal precedence implies temporal precedence in models of concurrency theory~\cite{%
%Casley-etal,
Gaifman1989}).

\item In~\cite{KellisonHsu2024}, a related graded monad, called the {\em graded neighbourhood} monad, was developed for a functional programming language {\tt Numerical Fuzz} whose type system can track quantitative bounds on roundoff error. 
The quantale $\V$ is now $(([0, \infty),\ge),+,0)$ (notice the greater-or-equal order). For a (generalised) metric space $(X,d)$, $T_rX$ has the same underlying set as in the previous example, but with different metric 
\[
d_{T_r X}((x_1,x_2),(y_1,y_2))= d_X(x_1,y_{1})
\]
Intuitively, each program is seen as a non-expansive map (in a metric semantics), and the monadic type $T_rX$ consists of pairs of values, where the first value is the (ideal) output of the program, while the second value is the actual output, under some finite approximation. The chosen metric on $T_rX$ carries then a natural interpretation: 
the distance depends only on how close the ideal results are, and completely ignores differences in their finite approximations, hence the denotational model is robust to approximation errors.
\end{enumerate}

\end{example}

Further examples can be found in the references cited above. We shall give more examples in Section~\ref{sec:M-acts}, dedicated to  those graded monads whose unit and multiplication are isomorphisms (that is, actions of monoidal categories in the usual sense).

\begin{remark}[Kleisli presentation of graded monads] Graded monads admit a Kleisli-like presentation~\cite{Katsumata2014} that we shall  recall below just for completeness: to give an $\M$-graded monad on a category $\A$ is the same as the following data: 
\begin{itemize}
\item For each $M$ of $\M$, a correspondence $X\longmapsto T_M X$ on objects of $\A$;
\item A unit morphism $\eta_X:X \to T_I X$ for every object $X$ of $\A$;
\item A Kleisli extension operator, mapping every morphism $f:X \to T_N Y$ of $\A$ to a family of morphisms $(f_M^\dagger: T_M X \to T_{M \otimes N} Y)_{M}$, natural in $M$;
\end{itemize}
subject to the following requirements:
\[
f_I^\dagger \circ \eta_X = f\ ; \ \id_{T_M X} = (\eta_X)_M^\dagger \ ; \ (g_N^\dagger \circ f)_M^\dagger = g_{M \otimes N}^\dagger \circ f_M^\dagger 
\]
where $f: X \to T_N Y$, $g: Y \to T_{M'} Z$.

\end{remark}

%======================================%

\subsection{Graded distributive laws of endofunctors over graded monads}\label{sec:dl}

A distributive law of an endofunctor $F:\A \to \A$ over a monad $T$ on $\A$ is a natural transformation $\lambda: FT \to TF$ compatible with the unit and the multiplication of the monad. 
There is a one-to-one correspondence between such distributive laws and liftings%
\footnote{Actually, these are {\em extensions}, not  {\em liftings}, although the latter term is already standard in the literature. See also the discussion in~\cite{ManesMulry2007}.} %
of $F$ to the Kleisli category of $T$~\cite{LPW2000}.

\medskip

In this section we look at distributive laws of endofunctors over graded monads. Later, in Section~\ref{sec:kl-lift} we shall see that  our definition of such distributive law is justified by the correspondence with Kleisli lifts (the Kleisli construction for graded monads being recalled in Section~\ref{sec:gr-kl}).

\medskip

More than $50$ years ago, B\'enabou characterised monads as lax functors from $\Sigma \one$ to $\cat$~\cite{Benabou1967}, where $\Sigma \one$ is the {\em suspension} of the category $\one$ (the one-object bicategory whose hom-category is $\one$), and $\cat$ is the $2$-category of categories, functors, and natural transformations. 
From this perspective, a distributive law $FT \to TF$ of an endofunctor $F$ over a monad $T$ is precisely a {\em colax} natural transformation $T \to T:\Sigma\one \to \cat$ between lax functors.\footnote{In~\cite{Street1972}, these were called {\em right lax transformations}.} 
All these generalises straightforward upon replacing $\one$ by an arbitrary monoidal category $\M$: there is an  induced one-object bicategory $\Sigma \M$, and a graded monad $T:\M \to [\A,\A]$ can be equivalently described as a lax functor $\Sigma \M \to \cat$, mapping the unique object of $\Sigma \M$ to $\A$, with action on the hom-category given by $T$. A colax natural transformation $T \to T:\Sigma\M \to \cat$ between lax functors is what we shall call a {\em graded} distributive law. Spelling this out yields the following definition:

\begin{definition}\label{defn:dl}
Let $T$ be an $\M$-graded monad on a category $\A$, and let $F$ be a functor on $\A$. 
A {\em graded distributive law} of $F$ over $T$ consists of a family of natural transformations $\lambda_M: F T_M \to T_M F$, natural also in the objects of $\M$, such that the following diagrams commute: 
\begin{equation}\label{eq:dl}
\begin{gathered}
\xymatrix@C=40pt{F T_M T_{M'} \ar[r]^{\lambda_{M}T_{M'} }\ar[d]_{F \mu_{M,M'} } 
&
T_M F T_{M'} \ar[r]^{T_M \lambda_{M'}} 
& T_M T_{M'}F  \ar[d]^{\mu_{M,M'}F}
\\
FT_{M \otimes M'} \ar[rr]^{\lambda_{M \times M'}} 
&&
T_{M \otimes M'}F
}
\\
\xymatrix{
& F \ar[dl]_{\eta_F} \ar[dr]^{F \eta} & 
\\
F T_I  \ar[rr]^{\lambda_I} && T_I F
}
\end{gathered}
\end{equation}
\end{definition}

In Section~\ref{sec:cst-functors} we shall focus on $\M $-actegories, rather than arbitrary graded monads. We shall then refer to functors endowed with such a distributive law as {\em costrong}, emphasising the duality with the strong functors used in functional programming.

\begin{remark}\label{rem:gen-dl}
\begin{enumerate}
\item The definition above allows for more generality, in the sense that one can consider not one but two $\M$-graded monads, seen as lax functors $T,S:\Sigma \M \to \cat$, and a lax natural transformation between these, which amounts to a pair $(F, (\lambda_M: FT_M \to S_M F)_M)$, where $F$ is a functor between the categories upon which the graded monads act, satisfying similar diagrams to those of Definition~\ref{defn:dl}. When $\M=\one$, such pairs were called {\em opmonad morphisms} in~\cite{Street1972}. 

One step further, the monoidal category itself can be varied. Then a distributive law between an $\M$-graded monad $T$ on $\A$ and an $\N$-graded monad $S$ on $\B$ is a triple $(\Phi, F,\lambda)$ where $\Phi:\M \to \N$ is a pseudo-monoidal functor between the categories of grades, $F:\A \to \B$ is a functor, and $(\lambda_{M}:FT_M \to S_{\Phi(M)} F_M)_M$ is a family of natural transformations required to verify the analogue of~\eqref{eq:dl}.
Even more generally, the two monoidal categories may be linked not by a pseudo-monoidal functor, but by a {\em span} of pseudo-monoidal functors $\Phi:\kat K\to \M$, $\Psi:\kat K \to \N$, in which case the graded distributive law takes the form $(\lambda_K: FT_{\Phi K} \to S_{\Psi K} F)_K$.

\item In~\cite{GKOBU2016}, distributive laws of a graded comonad over a graded monad are considered. 
Both the graded monad and comonad act on the same category $\A$, but the ordered monoids inducing the gradings are not the same. This issue is solved 
by a pair of monoid actions subject to a series of axioms ensuring, among all, that a(n ordered) monoid structure can be obtained on the cartesian product of these two monoids (hence a span of projections as described above). 

These distributive laws are at the same time more general and more specific than ours: more general in that they allow arbitrary categories of gradings related by a matched pair of actions, and more specific in that they involve distributivity of a (graded) comonad rather than a plain endofunctor over a graded monad (therefore compatibility is required between the distributive law and the (co)multiplication/(co)unit of the respective graded monad and graded comonad). 
\end{enumerate}
\end{remark}

An immediate consequence of Definition~\ref{defn:dl} is that graded distributive laws {\em compose}: 

\begin{proposition}\label{prop:compos-dl}
If $F,G:\A\to \A$ are functors which distribute over the $\M$-graded monad $T$ via  
$\lambda^F_M: FT_M \to T_M F$, respectively $\lambda_M^G: GT_M \to T_M G$, then the composite functor $FG$ distributes again over $T$ by 
\begin{equation}\label{eq:compositeDL}
\lambda^{FG}_M:\xymatrix{
	F G T_M 
	\ar[r]^{F\lambda_M^G}
	&
	F T_M G 
	\ar[r]^{\lambda_M^F G}&
	T_M F G}
\end{equation}
\end{proposition}

\begin{proof}
The result follows from the commutative diagram below:
\[
\xymatrix@C=40pt{
F G T_M T_{M'} 
\ar[r]^{F\lambda^G_{M}T_{M'} }
\ar[dd]_{FG \mu_{M,M'} } 
&
FT_M G T_{M'} 
\ar[r]^{\lambda^F_{M}GT_{M'}} 
& 
T_M F G T_{M'}  
\ar[r]^-{T_M F\lambda^G_{M'}}
&
T_M F T_{M'} G
\ar[r]^-{T_M \lambda^F_{M'}G}
&
T_M T_{M'} F G
\ar[dd]^{\mu_{M,M'}F G}
\\
&
&
F T_M T_{M'}G
\ar@{<-}[ul]^-{FT_M \lambda^G_{M'}}
\ar[ur]_-{\lambda^F_M T_{M'} G}
\ar[d]^{F\mu_{M,M'} G}
&
&
\\
F G T_{M \otimes M'} 
\ar[rr]^{F\lambda_{M\otimes M'}^G} 
&
&
F T_{M \otimes M'} G
\ar[rr]^-{\lambda_{M\otimes M'}^F G}
&
&
T_{M \otimes M'} F G
}
\]
and a similar one involving the unit of the graded monad. In the diagram above, the left and right pentagons commute by definition of the distributive laws of $G$, respectively $F$, while the top-middle square commutes by naturality. 
\end{proof}

\begin{remark}\label{rem:compos-dl}
Observe that, unlike distributive laws between (graded) monads (or between a monad and a comonad), no additional conditions are required for Equation~\eqref{eq:compositeDL} to define a valid graded distributive law. 
In the case of distributive laws between monads, the issue is that one wishes their composite to carry a monad structure as well, which occurs {\em only if} certain additional equations are satisfied. 
These equations ensure that the distributive law is compatible with the unit and the multiplicationof the monads.
\end{remark}

\begin{example}\label{ex:gr-dl}
\begin{enumerate}

\item \label{ex:gr-dl-maybe} Let $A$ be a fixed set of values and take $F:\set \to \set$ the functor $FX = \one + A \times X$, with $\one =\{0\}$. Then a distributive law of $F$ over the $\mon$-graded {\tt GMaybe} monad from Example~\ref{ex:gr-mnds}.\ref{ex:gr-maybe} can be defined as follows: $\lambda_{\mathtt f}$ is the unique function to the terminal set $\one$, $\lambda_{\mathtt s}$ is the identity on $F$, while $\lambda_{\mathtt m}: F\, \mathtt{GMaybe}_{\mathtt m} = \one + A \times (\one +-)  \to \mathtt{GMaybe}_{\mathtt m} \, F = \one + (\one+A \times -)$ maps $0$ (failure) to itself, pairs $(a,*)$ to $*$ (the pure value of type $A$ is dropped when there is no continuation), respectively pairs $(a,x)$ to themselves (value and state are preserved).

\item 
It is well-known that polynomial functors on $\set$ distribute over commutative  monads~\cite{HasuoJacobsSokolovaCMCS2006,ManesMulry2007}, the distributive law being recursively defined on the structure of the polynomials.\footnote{And sometimes called the {\em canonical} distributive law.} 
Now, if the monad is replaced by a graded monad $T:\M\to [\set,\set]$ is instead considered, we can see where the construction in~\cite{HasuoJacobsSokolovaCMCS2006} of a distributive laws generally fails. 
First, the component functors $T_M$ of the graded monad are not-necessarily pointed, hence if $F$ is a constant functor, there is no natural transformations $FT_M \to T_M F$. 
Second, there is no recipe on how to distribute the product $F\times G$ over each $T_M$, given distributive laws of $F$ and $G$. 
In~\cite{HasuoJacobsSokolovaCMCS2006}, this last issue is solved by recurring to {\em commutative} monads. 
In fact, a closer look at the recipe in~{\em op.cit.} shows that only a {\em lax monoidal} structure on the monad is necessary. 
To fix both issues, we can, for instance, consider graded monads whose components are {\em applicative} functors on $\set$. The previous example is precisely an instance of these {\em canonical} distributive laws.

\item Let $\V$ be a quantale and $F:\set \to \set$ a functor. Denote by $F_\V$ the extension of $F$ to $\mathbf{Met}(\V)$ which coincides with $F$ on discrete $\V$-metric spaces~\cite{BKV2019}. Assume that $F$ preserves weak pullbacks. Then $F_\V$ admits the following description using Barr's relation lifting~\cite{Barr1970}: for each $\V$-metric space $(X,d_X)$, $F_\V (X,d_X)$ is $FX$, with $\V$-distances~\cite{BKV2019} 
\begin{equation}\label{eq:V-metric}
%\begin{aligned}
%&
d_{F_\V (X,d_X)}(u_1,u_2)
%&&
= 
%&&
\bigvee_r \{r \mid \exists\, w \in F(T_r X) 
%\\
%&&&&&
%\mbox{ such that }  
\, . \, F\pi_{1}(w)=u_{1}
\, , \,
%\mbox{ and } 
F\pi_2 (w)=u_2 \}
%\end{aligned}
\end{equation}
%$F_\V$ acts as $F$ also on arrows. 
For example, if $F$ is the powerset functor, then the formula above shows that $F_\V$ carries the Hausdorff-Pompeiu metric on subsets. 
The $\V$-metric of~\eqref{eq:V-metric} naturally induces a distributive law of $F_\V$ over the graded monad of Example~\ref{ex:gr-mnds}.\ref{ex-rel-presh} by $\lambda_{r} : F_\V T_r  \to T_r F_\V $ given by
\[
w\in F_\V T_r (X,d_X) \mapsto \lambda_r(w)=(F\pi_1(w), F\pi_2(w))\in T_r F_\V  (X,d_X)
\]
\end{enumerate}
\end{example}

More examples of distributive laws will be provided in the remaining sections of the paper (e.g. Example~\ref{ex:cst-functors}, Section~\ref{sec:costrong=copointed}).

%======================================%

%======================================%

\subsection{The Kleisli construction for a graded monad~\cite{FujiiKatsumataMelies-FoSSaCS2016}}\label{sec:gr-kl}

Given a monad $T$ on a category $\A$, recall that the Kleisli category $\A_T$ associated to $T$ has the same objects as the base category $\A$, but an arrow $X \nto Y$ in $\A_T$ is in fact an arrow $X \to TY$ in $\A$. 
There is an adjunction $L_T:\A\rightleftarrows \A_T:R_T$, where the left adjoint $L_T$ is the identity on objects and postcomposition with the unit of the monad on arrows, while the right adjoint applies $T$ on objects; on arrows, $R_T$ acts as $T$ followed by the multiplication of the monad. 
This adjunction induces the monad $T$, in the sense that $T = R_TL_T$, and it is initial among all adjunctions inducing the monad. 

\medskip

In order to motivate the Kleisli construction for graded monads, that we shall recall a few lines below, Fujii, Katsumata and Melli\`es introduced the notion of 
{\em resolution} of a graded monad~\cite{FujiiKatsumataMelies-FoSSaCS2016}, consisting of an adjunction $L:\A \rightleftarrows \B:R$ and of a {\em strict} graded monad $S:\M\to [\B,\B]$. %
These data are required to induce the original graded monad, in particular $T_M = R S_M L$ holds for all $M$. 
When $\M=\one$, a graded monad on a category $\A$ is just a monad on $\A$. Observe that there is a {\em unique strict} $\one$-action on $\A$, namely the identity, hence a resolution in this case is just an adjunction inducing the monad.   
It was shown in~\cite{FujiiKatsumataMelies-FoSSaCS2016} that a graded monad $T$ may admit more than one resolution, and that the Kleisli construction, whenever it exists,\footnote{The morphisms in this Kleisli category (see the next paragraphs) are obtained using a coend, which in some cases might not exist.} is the {\em initial resolution} of $T$. 

\medskip

The Kleisli category $\A_T$ for an $\M$-graded monad $T$ on a category $\A$ is obtained as follows~\cite{FujiiKatsumataMelies-FoSSaCS2016}:
The objects of $\A_T$ are pairs $(M,X)$, where $M$ is an object of $\M$ and $X$ is an object of $\A$. 
Intuitively, a pair $(M,X)$ represents ``values of type $X$ in a context/grade $M$'', like effect bound. 
The hom-sets of $\A_T$ are given by the coend formula
\begin{equation}\label{eq:kl-coend}
\A_T((M,X),(M',X')) = \int^{N \in \M} \A(X,T_N X')\times \M(M\otimes N,M')
\end{equation}
Explicitly, an arrow $(M,X)\to(M',X')$ is an equivalence class of triples $(N, X \overset{f}{\to} T_N X', M \otimes N\overset{u}{\to} M')$ modulo the equivalence relation generated by 
\[
%\begin{align*}
(N, X \overset{f}{\to} T_N X', M \otimes N\overset{\id \otimes v}{\to}M \otimes N'\overset{u}{\to} M') 
\\
\sim 
(N', X \overset{f}{\to} T_N X'\overset{T_r \id}{\to} T_{N'} X', M \otimes N'\overset{u}{\to}{M'})
%\end{align*}
\]
for every $v:N \to N'$ in $\M$. The object $N$ can be perceived as the internal (hidden) grade for the computation, $f:X \to T_N X'$ 
is a morphism with graded/effectful output, while $u:M \otimes N \to M'$ represents grade correction. 
Intuitively, in programming, one starts with a pure input $X$ (in context $M$), then performs a computation $f$ that uses the additional grade $N$ to produce $X'$ (wrapped in the $N$-effect). 
Finally, the total grade $M \otimes N$ is adjusted to match the output grade $M'$ via $u$ (e.g., $u$ could be an isomorphism for an exact match, or a morphism rellabeling or discarding resources). 
The coend ``integrates'' over all possible intermediate $N$, making thus morphisms flexible.

The identity morphism on an object $(M,X)$ is induced by  the triple $(I, X \overset{\eta_X}{\to} T_I X, M \otimes I \cong M)$, while composition of morphisms is defined via the universal property of coends, which described explicitly gives 
\begin{align*}
(N', X' \overset{f'}{\to} T_{N'} X'', M' \otimes N'\overset{u'}{\to} M'')\circ (N, X \overset{f}{\to} T_N X', M \otimes N\overset{u}{\to} M')
\\
=
(N \otimes N', 
X \overset{f}{\to} T_N X' \overset{T_{\id} f'}{\to} (T_N \circ T_{N'}) X'' \overset{\mu}{\to} T_{N \otimes N'} X'', 
\\
M \otimes (N\otimes N') \cong (M \otimes N) \otimes N' \overset{u\otimes \id }{\to} M'\otimes N' \overset{u'}{\to} M''
)
\end{align*}

There is a strict $\M$-graded monad $S$ on $\A_T$\footnote{Technically, if the tensor product on $\M$ is not strictly associative and unital, this is only a {\em pseudo}-action, but this mild inconvenience does not obstruct our results.} induced by the regular action of $\M$ on itself as $S_M (N ,X) = (M\otimes N,X)$ 
and an adjunction
\[
\xymatrix@C=40pt{\A \ar@<1ex>[r]^{L_T} \ar@{}[r]|{\perp} \ar@{<-}@<-1ex>[r]_{R_T}& \A_T }
\]
where the left adjoint $L_T$ maps an object $X$ to $(I,X)$, while $R_T(M,X)= T_M X$. These determine a resolution of the original graded monad by 
\[
(R_T  S_M L_T)X = T_{M\otimes I}X \cong T_M X
\]
and in~\cite{FujiiKatsumataMelies-FoSSaCS2016} it was shown to be initial resolution of the $\M$-graded monad $T$.

\begin{remark}
There are two potential issues regarding the Kleisli construction outlined above: 
Firstly, the existence of the coends used for the homsets. 
Secondly, the associativity of the composition and the behaviour with respect to the identity arrow, which hold only up to isomorphisms. 
We shall discuss the first one, but let us briefly note that the second issue can be treated in the simplest way by restricting to a {\em strict monoidal} category $\mathcal M$; for example, one may take $\mathcal M$ to be a partially ordered monoid.\footnote{This is the situation addressed in most of the existing literature on graded monads in computation.}
Returning to the first issue, observe that the coend of~\eqref{eq:kl-coend} certainly exists when $\M$ is small, but it does also in other familiar situation, namely when the monoidal category $\M$ is (right) closed: 
\[
\M(M \otimes N,P) \cong \M(N,[M,P])
\]
In this case, Yoneda lemma shows that Equation~\eqref{eq:kl-coend} becomes
\begin{align}
&
\A_T((M,X),(M',X'))
&&
=
&&
\int^{N \in \M} \A(X,T_N X')\times \M(M\otimes N,M')
\nonumber
\\
&&&
\cong 
&&
\int^{N \in \M} \A(X,T_N X')\times \M(N,[M,M'])
\nonumber
\\
&&&
\cong
&&
\A(X,T_{[M,M']} X')\label{eq:kl-coend-resid}
\end{align}
Therefore an arrow $(M,X)\to (M',X')$ in $\A_T$ is an arrow $X \to T_{[M,M']} X'$ in $\A$, with composition induced by the internal hom composition in $\M$: 
\[
[M,M'] \otimes [M',M''] \to [M,M'']
\]
Taking the mate of the identity arrow on $M$, namely $I \to [M,M]$ produces the identity arrow on $(M,X)$ in $\A_T$ as:
\[
\id_{(M,X)}: \one \overset{\id_X}{\to} \A(X,X) \overset{\eta_X}{\to} \A(X,T_I X) \overset{}{\to} \A(X,T_{[M,M]} X)
\] 
In the coend presentation~\eqref{eq:kl-coend} of hom-sets of $\A_T$, to give a morphism $(M,X) \to (M',X')$ assumes an internal  choice of an intermediate grade $N$, producing an $N$-effectful computation $X \to T_N X'$ while separately giving a grade transformer $M \otimes N \to M'$. When $\M$ is right closed, instead of choosing $N$ and them mapping $M \otimes N$ to $M'$, a direct effect, with grade $[M,M']$, is available.  
\end{remark}

\begin{example}\label{ex:dl}
\begin{enumerate}

\item\label{ex:dl-GMaybe} (Graded Maybe) The Kleisli category of the graded Maybe monad from Example~\ref{ex:gr-mnds}.\ref{ex:gr-maybe} has as objects pairs $(\mathtt{a},X)$ where $\mathtt a\in \mon$ indicates the effect of a (previous) computation, and $X$ is a set of pure values. Observe that the binary operation on $\mon$ is only partially residuated, that is, the internal hom in $(\mon, \otimes, \mathtt s)$ exists with two exceptions, namely $[\mathtt f, \mathtt s]$ and $[\mathtt m, \mathtt s]$, being given by 
\[
\begin{array}{|c|c|c|c|}\hline
[-,-] & \mathtt f & \mathtt s & \mathtt m
\\
\hline
\mathtt f & \mathtt m & \not\exists & \mathtt m 
\\
\hline
\mathtt s & \mathtt f & \mathtt s & \mathtt m  
\\
\hline
\mathtt m & \mathtt f & \not\exists & \mathtt m  
\\
\hline
\end{array}
\] 
Although the internal hom is not everywhere defined, we can still use Equation~\eqref{eq:kl-coend-resid} whenever $[-,-]$ exists (as it does not involve the coend parameter) and obtain that the hom-sets of the Kleisli category for the graded Maybe monad are as follows: 
\[
\begin{array}{|c||c|c|c|}
\hline
\set_{\tt GMaybe}((-,X),(-,Y)) & (\mathtt f,Y) & (\mathtt s, Y) & (\mathtt m,Y) 
\\
\hhline{|=|=|=|=}
(\mathtt f,X) & \set(X,\one + Y) &  & \set(X,\one+Y)
\\
\hline 
(\mathtt s,X) & \one & \set(X,Y) & \set(X,\one+Y)
\\
\hline
(\mathtt m,X) & \one &  & \set(X,\one+Y)
\\
\hline
\end{array}
\]
For the missing cases in the table above, we computed the coend~\eqref{eq:kl-coend} by the usual formula involving coproducts and coequalisers, and obtained that
\[
\set_{\tt GMaybe}(({\mathtt f},X),({\mathbf s},Y)) = \emptyset \ , \ \set_{\tt GMaybe}(({\mathtt m},X),({\mathtt s},Y)) = \emptyset
\]
The fact that there are no arrows in $\set_{\tt GMaybe}$ between $(\mathtt{f},X)$ and $(\mathtt{s},Y)$ supports the intuition 
that a computation intended to produce values of type $X$ but which had failed cannot subsequently yield a successful value of type $Y$.
Similarly, $\set_{\tt GMaybe}((\mathtt{m},X),(\mathtt{s},Y)) = \emptyset$ indicates that it is impossible to guarantee a successful computation of a pure value of type $Y$ when starting from a computation in $X$ that may fail.

\item Another example is obtained when taking $\M$ to be $\set$, with cartesian monoidal structure, and the graded {\tt Writer} monad induced by the regular action of $\set$ upon itself, i.e. $T_M X = M \times X$. 
In the Kleisli category of the (ungraded) {\tt Writer} monad, each computation produces a result (a pure value) and  a log-type message. Kleisli composition automatically accumulates the logs in the order they appeared, alongside the main computation result, without having the computation depending on previous logs (it is ``write-only'').
Turning to the graded version $T$ of the {\tt Writer} monad, observe that $\set$ is cartesian closed, hence the Kleisli category of $T$ has as objects pairs of sets $(M,X)$, with maps $X \to [M,M']\times X'$ as arrows $(M,X)\to (M',X')$. 
Here $[-,-]$ denotes the usual internal hom in $\set$ (that is, the {\em set} of functions $M \to M'$). 
Intuitively, a function $X \to [M,M']\times X'$ is a computation that takes a value of type $X$, produces a value of type $X'$, and provides a ``grade transformer'' $M \to M'$ depending only on the input value. 

\end{enumerate}
\end{example}

%======================================%

\subsection{Distributive laws and extension to Kleisli category}\label{sec:kl-lift}

Distributive laws $FT \to TF$ of an endofunctor $F$ over a monad $T$ on a category $\A$ are in one-to-one correspondence with liftings of $F$ to the Kleisli category $\A_T$ of the monad, that is, with endofunctors $\tilde F:\A_T \to \A_T$ satisfying $\tilde F L_T = L_T F$:
\[
\xymatrix{\A_T \ar[r]^{\tilde F} & \A_T 
\\
\A \ar[u]^{L_T} \ar[r]^F & \A \ar[u]_{L_T}
}
\]

We shall see how this correspondence extends from monads to graded monads, using the distributive laws introduced in Section~\ref{sec:dl}. 

First, consider a distributive law $\lambda_M: FT_M \to T_M F$ of a functor $F$ over an $\M$-graded monad $T$ on a category $\A$. Construct a functor $\tilde F$ on $\A_T$ as follows: on objects, put 
\begin{equation}\label{eq:kl-lift-obj}
\tilde F(M,X)=(M, FX)
\end{equation}
while on arrows use the distributive law: 
\begin{equation}\label{eq:kl-lift-arrows}
\begin{aligned}
& \A_T((M,X),(M',X'))&& = &&\int^{N \in \M} \A(X,T_N X')\times \M(M\otimes N,M')
\\
&&&\overset{F}{\to} &&\int^{N \in \M} \A(FX,FT_N X')\times \M(M\otimes N,M')
\\
&&&\overset{\lambda}{\to} &&\int^{N \in \M} \A(FX,T_N FX')\times \M(M\otimes N,M')
\\
&&&= &&\A_T((M,FX),(M',FX'))
\end{aligned}
\end{equation}
Observe that $\tilde F$ (strictly) distributes over the $\M$-graded monad $S$ on $\A_T$ by 
\[
\tilde F S_M (N,X) 
%= \tilde F (M \otimes N,X )  
= (M \otimes N, FX) 
%= S_M (N,FX) 
= S_M \tilde F(N,X)
\]
Also, $\tilde F L_T=L_T F$ holds, hence $F$ lifts to $\A_T$. 

\medskip

Unlike liftings of ordinary (ungraded) monads, graded monads require that the lifting of the functor also takes the grading into account.
We shall now see that this requirement is essential for establishing the converse, from liftings back to distributive laws.\footnote{In fact, unravelling the $2$-categorical constructions of~\cite{FujiiKatsumataMelies-FoSSaCS2016} used to justify the Kleisli (and Eilenberg-Moore) categories for a graded monad, we can see that a distributive law as above corresponds to a genuine opmonad morphism, but {\em over} the (pseudo)comonad $[\M,-]:\cat\to \cat$.}
Consider thus a lifting $\tilde F$ of $F$ to $\A_T$, together with a distributive law $(\tilde \lambda_M:\tilde F S_M \to S_M \tilde F)_M$. We use that $\A_T$ is a resolution of the graded monad in order to define $\lambda_M: FT_M \to T_M F$ as 
\[
FT_M = F R_T S_M L_T \to R_T \tilde F S_M L_T \overset{\tilde \lambda_M}{\to} R_T S_M \tilde F L_T = R_T S_M L_T F = T_M F
\] 
where $F R_T  \to R_T \tilde F$ is the mate of $\tilde F L_T=L_T F$. Verification of diagrams~\eqref{eq:dl} follows from those of $\tilde F$.

\begin{remark}\label{rem:kl-bicat}
\begin{enumerate}
\item In~\cite{FujiiKatsumataMelies-FoSSaCS2016}, the Kleisli and the Eilenberg-Moore constructions associated to an $\M$-graded monad $T$ were justified by constructing suitable $2$-categories of (co)lax cones in which $T$ was actually a monad. 
Using those constructions, we can see that our notion of a distributive law of an endofunctor over the graded monad turns to be the usual distributive law of a $1$-cell over a monad (within the respective $2$-category). 
However, in order to avoid the $2$-categorical machinery, we chose the direct approach. 

\item \label{rem:kl-bicat2} The appearance of the coend in the definition of the Kleisli category makes the machinery exhibited in~\cite{FujiiKatsumataMelies-FoSSaCS2016} delicate if $\M$ is not small nor strict. In fact, even the notion of an $\M$-graded monad is in itself of a bicategorical nature.
Hence it is more suitable to think of $\A_T$ as being a bicategory, with objects as above, $1$-cells $(M,X)\to (M',X')$ being triplets 
\[
(N, X \overset{f}{\to} T_N X', M \otimes N\overset{u}{\to} M')
\]
and $2$-cells $(N,f,u)\to (P,g,v):(M,X)\to (M',X')$ given by $w:N \to P$ such that $T_w \id_{X'} \circ f = g$ and $v \circ (\id_M \otimes w) = u$. 
The Kleisli category of~\cite{FujiiKatsumataMelies-FoSSaCS2016} is then recovered by applying the connected components functor to each hom-category~\cite{Benabou1967}. 

\end{enumerate}
\end{remark}

%======================================%

\section{Actegories and Costrong Functors}\label{sec:M-acts}

Let again $(\M,\otimes,I)$ be a monoidal category. 
Rewriting an $\M$-graded monad $T: \M \to [\A,\A]$ in the uncurried form $\M\times \A \to \A$, we may as well interpret it as a (lax) action of the monoidal category $\M$ on $\A$, as in~\cite{Benabou1967}. 
In particular, the two natural transformations $\eta:\id_\A \to T_I$, $\mu_{M,N}:T_M  T_N \to T_{M\otimes N}$, witnessing the lax monoidal structure of $T$, rewrite as
\[
\eta_X: X \to I \cdot X \ , \ \mu_{M,N,X}: M \cdot (N \cdot X) \to (M \otimes N)\cdot X
\]
where we denoted $M\cdot X$ instead of $T_M(X)$ to emphasise the action. 
From now on, we shall assume that $\eta$ and $\mu$ are isomorphisms, and talk about {\em actions} of the monoidal category $(\M,\otimes, I)$, or $(\M,\otimes, I)$-actegories~\cite{CapucciGavranovic2022,McCrudden2000,Pareigis1977}, instead of $\M$-graded monads.
Whenever the monoidal product on $\M$ is clear from the context, we shall refer to them as $\M$-actegories.
If $\eta$ and $\mu$ are identities, we shall call the action {\em strict.}
By~\cite{McCrudden2000}, any action of a monoidal category is equivalent to a strict one. Consequently, in the sequel we shall either ignore the natural isomorphisms $\eta, \mu$ and write as if the action was strict, or write these isomorphisms unlabelled if no confusion may arise.

\begin{example}\label{ex:actegories}

	\begin{enumerate}
	
		\item \label{ex:action=own tensor}  
		Any monoidal category $\M$ canonically acts on itself via the tensor product. We shall call this action the {\em regular action} of $\M$. 
		In particular, this is the case when $\M$ is cartesian or cocartesian, the tensor being given by the product, respectively by the sum.

		\item \label{ex:action=functor evaluation}
            Let $\A$ be an arbitrary category. 
            There is a strict action on $\A$ of the monoidal category $([\A,\A], \circ, \id_\A)$ given
by evaluation, sending $(F, X)$ to $FX$. In fact, the category of endofunctors $([\A,\A], \circ, \id_\A)$ is terminal in the $2$-category of monoidal categories acting on $\A$.\footnote{Formally, this $[\A,\A]$-actegory structure of $\A$ {\em classifies} actions on $\A$.}

		\item \label{ex:action=enrichment} 
		Let $\A$ be an $\M$-enriched category~\cite{Kelly-book}. If for any object $X$ of $\A$, the $\M$-valued hom functor $\A(X,-):\A \to \M$ admits a left adjoint $-*X$, the enriched category $\A$ is said to have {\em tensors} (or {\em copowers}) and the correspondence $(M,X) \mapsto M *X$ is an action of $\M$ on $\A$~\cite{JanelidzeKelly2001}. 
				
		\item \label{ex:action=(co)power} Any category with finite copowers\footnote{The copower of an object $X$ by a natural number $n$ is the coproduct of $n$-many copies of $X$.}
        is an $((\mathbb N,\le),+,0)$-actegory, with action given by copower
        \[
        (n,X) \mapsto n \cdot X = \coprod_n X
        \]
		Dually, any category with finite powers is an $((\mathbb N,\ge),\cdot,1)$-actegory with action
		\[
		(n,X) \mapsto X^n = \prod_n X
		\]		
		
		\end{enumerate}

\end{example}

More examples of actions of monoidal categories can be found in e.g.~\cite{CapucciGavranovic2022}.

%======================================%

\subsection{(Co)Strong functors as morphisms of actions of monoidal categories}\label{sec:cst-functors}

%Let $(\M,\otimes, I)$ be a monoidal category. 
In Section~\ref{sec:gr-mnd} we studied endofunctors which distribute over $\M$-graded monads. Here we shift to a more algebraic perspective: instead of distributive laws of functors over $\M$-graded monads, we talk about functors that respect or preserve the $\M$-action, and we adopt the terminology more commonly used in the literature.

\medskip

A functor $F:\M\to \M$ is called {\em strong} if it is equipped with a natural transformation $\st: - \otimes F(=) \to F(- \otimes =)$, called {\em strength}, subject to certain axioms expressing compatibility between the tensor product and the action of the functor~\cite{Kock1972}. 
While strong functors have origins in enriched category theory, they have since become central in the semantics of functional programming, modelling e.g.~computational effects or abstract syntax~\cite{FiorePlotkinTuri1999,Moggi1989}. 
The paper~\cite{McDermottUustalu2022} provides an overview on the theory of strong functors and strong monads.   
The notion of a strong functor admits a natural two-fold generalisation: first, extending from functors on \(\M\) to functors on \(\M\)-actegories; and second, considering functors with possibly distinct domain and codomain. 

\medskip

Just for convenience, let us properly say that an {\em $\M$-strong functor}\footnote{Also called {\em morphism of $\M$-actegories}, or {\em $\M$-linear functor}.} between two $\M$-actegories $\A$ and $\B$ is a functor $F:\A \to \B$ equipped with a natural transformation $\st:M\cdot F(X) \to F(M\cdot X)$, called {\em strength}, satisfying 
\begin{equation}\label{eq:lax-morphism}
\xymatrix@C=1pt@R=20pt{
(M\otimes N)\cdot F(X) 
\ar[d]_{\st}
\ar[rr]^{\cong}
&
&
M \cdot (N \cdot F(X))
\ar[d]^{\id \cdot \st}
&&
I \cdot F(X) 
\ar[rr]^{\st}
\ar[dr]_{\cong}
&
&
F(I \cdot X)
\ar[dl]^{\cong}
\\
F((M \otimes N)\cdot X)
\ar[dr]_{\cong}
&
&
M \cdot F(N \cdot X)
\ar[dl]^{\st}
&&&
F(X)
&
\\
&
F(M \cdot (N \cdot X))
}
\end{equation}

Dually, an $\M$-costrong functor is a functor $F:\A \to \B$ such that $F\op$ is a strong $\M$-functor from $\A\op$ to ${\B}\op$. Explicitly, $F$ is equipped with a natural transformation in the opposite direction, denoted $\cst:F(M\cdot X) \to M\cdot F(X)$ and called {\em costrength}, satisfying diagrams dual to those in~\eqref{eq:lax-morphism}:
\begin{equation}\label{eq:colax-morphism}
\xymatrix@C=1pt@R=20pt{
F(M \cdot (N\cdot X)) 
\ar[d]_{\cong}
\ar[rr]^{\cst}
&
&
M \cdot F(N \cdot X)
\ar[d]^{\id \cdot \cst}
&&
F(I \cdot X)
\ar[rr]^{\cst}
\ar[dr]_{\cong}
&
&
I \cdot F(X) 
\ar[dl]^{\cong}
\\
F((M \otimes N)\cdot X)
\ar[dr]_\cst
&
&
M \cdot (N \cdot F(X))
\ar[dl]^{\cong}
&&&
F(X)
&
\\
&
(M \otimes N) \cdot F(X)
}
\end{equation}

Intuitively, if the strength of an $\M$-strong functor {\em pushes} the action inside the functor, $\M$-costrong functors have the opposite behaviour, namely they {\em pull out} the $\M$-action~\cite{HughesVollmerOrchard2021}. This makes costrong functors interesting, among others, for various coinduction principles (see Section~\ref{sec:costrong-stream}).
 
\begin{remark}
In fact, what we now call {\em costrength} is nothing but the distributive law of Section~\ref{sec:gr-mnd} of $F$ over the $\M$-graded monad inducing the $\M$-action (in the case where $\A=\B$ and the actions coincide). 
The shift in terminology and notation is intended to parallel the theory of strong functors and strong monads, emphasizing the role of costrength as the dual counterpart of strength and making the analogy between the two settings more transparent.
\end{remark}

Despite the centrality of strength in applications, the dual notion of costrength, and more generally costrong functors, have received comparatively little attention in the literature. This asymmetry is somewhat surprising, especially given that many categorical phenomena such as comonadic structures, coalgebras, and context-dependent computation  naturally call for a dual treatment. 
In the context of a programming language, a functor provides a container to execute computation within it, but often we want to access results of a chain of computation directly.
Costrength, like the comonad counit, provides such means; however, in the case of costrength this is more nuanced, since it depends on the action with respect to which the functor is costrong.

\begin{remark}\label{ex:st}
As mentioned at the beginning of this section, strong functors first appeared in enriched category theory. 
If $\M$ is monoidal closed and self-enriched via the internal hom, then an endofunctor on $\M$ is $\M$-strong if and only if it is an $\M$-enriched functor~\cite{Kock1972}. 
In particular, if $(\M,\otimes,I)=(\set, \times, \one)$, the cartesian closed category of sets and functions, then {\em any functor} on $\set$ is automatically $\set$-enriched, hence it is $(\set, \times, \one)$-strong and its strength is {\em unique}. 
This explains the large popularity that strong functors gained in programming: they formalise the idea that a data-type constructor (a functor) can be applied not only on objects but also on the internal hom structure (the objects representing ``mappings between objects''). 

The bare dualisation of the above does not go as smoothly as one might expect:%
\footnote{See for example~\cite{ChoudhuryGay2025}, where dualisation of lambda-calculus is discussed.} 
First, {\em coclosed} monoidal categories are much rarer than closed ones, apart from the trivial example obtained by taking the dual of a closed monoidal category. 
In particular, this is the case for the monoidal products that typically arise in the semantics of programming languages.%
\footnote{We should, however, mention some well-known examples: 
First, the Klesli category of the continuation monad on $\set$ is cocartesian coclosed. 
Next, $*$-autonomous categories like sup-lattices are coclosed (with respect to the second tensor product).
In particular, compact closed categories, like finite dimensional vector spaces, are both closed and coclosed as monoidal categories. 
Thirdly, the category of polynomial $\set$-functors, endowed with functor composition as monoidal product, is coclosed -- the cohom being given by left Kan extension~\cite{NiuSpivak}.} %  
Second, to the best of authors' knowledge, there is no developed theory of {\em co}enriched categories.  

\end{remark}

\begin{remark}
The notions of {\em costrong functor} and {\em costrength} had first been considered, to the best of the authors' knowledge, in~\cite{BluteCockettSeely1997}, in the framework of (co)context categories. 
They subsequently had scarce appearances in the literature, usually focused on a specific actegory, e.g.~\cite{HansenKlin2011,HeunenKarvonen2015,McDermottUustalu2022}, and without a comprehensive analysis of their properties or potential applications. 
\end{remark}

A word of caution: the terms {\em costrong functor} and {\em costrength} have also been used to describe endofunctors on monoidal categories equipped with {\em right strengths} $F(-)\otimes = \, \to F(-\otimes =)$. Given that such functors are not considered here, we believe there will be no confusion.

\medskip

An {\em $\M$-strong natural transformation} between $\M$-strong functors $(F,\st),(G,\st): \A \to \B$ is a natural transformation $\alpha:F \to G$ satisfying 
\[
\xymatrix{
M \cdot F(X)
\ar[r]^{\id \cdot \alpha}
\ar[d]_\st
&
M \cdot G(X)
\ar[d]^\st
\\
F(M\cdot X) 
\ar[r]_\alpha
&
G(M\cdot X)
}
\]
$\M$-costrong natural transformations between $\M$-costrong functors are defined similarly, by reversing the direction of the vertical arrows in the diagram above.

\medskip

\begin{remark}
One important aspect of $\M$-(co)strong functors and $\M$-(co)strong natural transformations, is their {\em compositionality}. 
Although it comes {\em for free} from the perspective of category theory (see Proposition~\ref{prop:compos-dl} and Remark~\ref{rem:compos-dl}), compositionality is one of the key features in the semantics of computation. 
Formally, $\M$-actions, $\M$-costrong functors and $\M$-costrong natural transformations determine a $2$-category $\mact_\cst$~\cite{BluteCockettSeely1997}. 
In particular, $\mact_\cst(\A,\B)$ denotes the hom-category of $\M$-costrong functors and costrong natural transformations between two $\M$-actegories $\A$ and $\B$. 
As it is the case in any $2$-category, $\M$-costrong endofunctors determine a {\em monoidal} category $(\mact_\cst(\A,\A), \circ, \id)$ with respect to functor composition. 
\end{remark}

\begin{example}\label{ex:costrong}
	Let $\A$ be an $\M$-actegory. 
	In case $\M$ is braided (or even symmetric) monoidal, the action of an object $M$ of $\M$ on $\A$ determines a simultaneously $\M$-strong and $\M$-costrong {\tt Writer}-like functor $F_M:\A \to \A$, given by $F_M (X) = M \cdot X$ (the (co)strength being induced by the braiding/symmetry).%
	\footnote{It is the endofunctor part that we intend to emphasise in this section, not the (graded) monad aspect.} %
	If additionally $M$ is a (co)monoid in $\M$, then $F_M$ becomes a (co)monad, with associated (co)unit and (co)multiplication being $\M$-(co)strong natural transformations.
	That is, $F_M$ is an $\M$-(co)strong (co)monad.%
	\footnote{By an $\M$-(co)strong (co)monad we mean one for which the (co)multiplication and the (co)unit are compatible with the (co)strength in the obvious way.} %
	In particular, this is the case for cartesian monoidal categories acting on themselves, like $(\set,\times,\one)$, where every object $M$ naturally carries a comonoid structure, hence induces the $(\set,\times,\one)$-strong and costrong comonad $F_M (X) = M \times X$. 
	As mentioned above, the costrength $\cst : N \times (M \times X)  \to M \times (N \times X)$ is determined by the symmetry of the cartesian product. 
	The results of the subsequent Section~\ref{sec:costrong=copointed} will show that this is {\em the only $(\set, \times, \one)$-costrength} on the {\tt Writer} comonad (even on the underlying functor part). 
	
\end{example}

More examples of familiar functors on $\set$, costrong with respect to various monoidal actions, can be found below. 
In all these cases, whenever the functor in question is a (co)monad, both the (co)unit and the (co)multiplication are costrong:

\begin{example}\label{ex:cst-functors}
	\begin{enumerate}
	
		\item Consider first $(\set, \times, \one)$ acting on itself. 
        \begin{enumerate}
            \item For the {\tt Reader} monad $X \mapsto [S,X]$, $(\set,\times, \one)$-costrengths $[S,M\times X] \to M \times [S,X]$  are in one-to-one correspondence with elements of $S$.\footnote{For a formal argument, see Section~\ref{sec:costrong=copointed}.} 
            %
            %In particular, we see being costrong is a structure, not a property. 
            
            \item \label{ex:costrong-comonad-not-unique-cst}
                Composing the $(\set, \times, \one)$-costrong {\tt Writer} comonad from  Example~\ref{ex:costrong} with the {\tt Reader} monad above, we obtain $(\set, \times, \one)$-costrengths (again not unique!) on the {\tt State} monad $X \mapsto [S,S\times X]$, respectively on the {\tt Costate} comonad $X \mapsto S \times [S,X]$. 
                This is one instance where the compositionality of costregths proves useful.             

            \item To provide a class of negative examples, observe that $\set$-functors $F$ such that $F\emptyset \neq \emptyset$ cannot possess a $(\set, \times, \one)$-costrength. 
        \end{enumerate}

        \item Take now $(\set,+,\mathbb 0)$, again acting on itself. 
        \begin{enumerate}
        
        	\item The powerset functor $\mathcal P$ is $(\set,+,\mathbb 0)$-costrong, where $\cst:\mathcal P(M+X) \to M+\mathcal P(X)$ forgets all possible occurrences of elements of $M$ in a subset of $M+X$, resulting thus a subset of $X$, which subsequently embeds in $M+\mathcal P(X)$. 

            \item The above example generalises to the class of {\em filterable} functors: 
            a filterable functor $F$ is a functor with a natural transformation $\phi: F {\tt Maybe}\rightarrow F$. 
            The costrength is given by $F(M + X) \xrightarrow{!+\id} F(\mathbb 1 + X) \xrightarrow{\phi} F X \xrightarrow{\mathrm{inr}} M + FX$. 
            By extension, every functor composed with a filterable functor on the left produces another filterable functor, hence $(\set,+,\mathbb 0)$-costrong.

            \item Because $\set$ is a distributive category, we obtain for free that the {\tt Writer} comonad is also $(\set,+,\mathbb 0)$-costrong, with costrength $S \times (M + X) \cong (S \times M) + (S \times X) \xrightarrow{} M + (S \times X)$, where the last arrow is induced by the projection on the second component.

        \end{enumerate}

        \item Consider the action of $(\set\op,\times, \one)$ on $\set$ given by exponentiation. 
Since every functor $F:\set \to \set$ is (uniquely) $(\set, \times, \one)$-strong, taking the mate of its strength $M \times F(-) \to F(M\times -)$ shows that $F$ has a unique $(\set\op,\times, \one)$-costrength $\cst:F [M,X]\to [M,FX]$. Explicitly, $\cst (t)(m) = F(\mathsf{ev}_m(t))$, where $\mathsf{ev}_m:[M,X]\to X$ is given by evaluation $\mathsf{ev}_m(f) = f(m)$~\cite{HansenKlin2011}.

		\item Consider an arbitrary category $\A$ with copowers and the strict monoidal category $([\A,\A],\circ, \id_\A)$ acting on $\A$ via functor application. 
		Copowering with a fixed set gives a $([\A,\A],\circ, \id_\A)$-costrong functor $F:\A \to \A$, $F(X) = S \cdot X$. The costrength now represents a distributive law of $F$ over {\em all functors}. Its existence is guaranteed by the fact that every functor {\em laxly preserves colimits}, in particular copowers. 		

		\item \label{ex:traversable}

		Finally, let us consider the strict monoidal category $[\set,\set]_{\mathsf{appl}}$ of {\em applicative} functors and {\em applicative} natural transformations~\cite{McBridePaterson2008}, acting on $\set$ by evaluation. Then a $[\set,\set]_{\mathsf{appl}}$-costrong functor is known as {\em traversable}~\cite{JaskelioffRypacek2012}. 

		Traversables can be equivalently characterised as Eilenberg-Moore coalgebras for the comonad on $[\set,\set]$ mapping a functor $T$ to the polynomial $\sum_n T(n)\times \id^n$. Consequently, if $T$ is traversable, then for any set $X$, an element $t\in T X$ is uniquely determined by a unique $n\in \mathbb N$, unique values $t_n \in T(n)$ and $x \in X^n$ such that $t = T(x)(t_n)$. 

	\item One can consider a variant of the last example above, by looking at the canonical action of the strict monoidal category of {\em all} functors $([\set,\set],\circ, \id_\set)$ on $\set$. 
	Repeating with mild adjustments the reasoning for traversable functors from~\cite{Roman2019}, we see that the category of $([\set,\set],\circ, \id_\set)$-costrong functors on $\set$ is isomorphic to the Eilenberg Moore category of coalgebras for the comonad on $[\set,\set]$ given by $T \mapsto T(\one) \times \id$. Hence each $t\in TX$, for a set $X$, is uniquely determined by values $t_\one\in T(\one)$ and $x\in X\cong \set(\one,X)$ such that $t=T(x)(t_\one)$. In particular, each $([\set,\set],\circ, \id_\set)$-costrong functor comes equipped with a copoint $T \to \id$. 
As expected, these $([\set,\set],\circ, \id_\set)$-costrong functors are significantly more restrictive than the traversable ones. 

	\end{enumerate}
\end{example}

A simple but efficient method to produce monoidal actions is by considering natural numbers acting on various endofunctor categories:

\begin{example}
\begin{enumerate}
	\item Let $F:\A \to \A$ be a functor. The discrete additive monoid $(\mathbb N, +, 0)$ strictly acts on $[\A, \A]$ by $n \cdot_F G = F^n G$. This graded monad has been considered e.g. in~\cite{MiliusPattinsonSchroder-calco2015}.
	
	\item If we want to extend the above action to one of the {\em ordered} additive monoid $((\mathbb N, \le), +, 0)$, more data is necessary, as follows: 
	Let now $F:\A \to \A$ be a {\em well-pointed} functor, that is, a functor endowed with a natural transformation $\eta: \id \to F$ (a {\em point} of $F$) such that $F\eta = \eta F$~\cite{Kelly1980}.
	Then the previous formula $n \cdot_F G = F^n G$ produces now an $((\mathbb N, \le), +, 0)$-action on $[\A, \A]$. 
	Observe that the order relation $m \le n$ is witnessed by a natural transformation $F^m G \to F^n G$, which is {\em well}-defined due to $F\eta = \eta F$. 	
	\end{enumerate}
\end{example}

In view of the above examples, a functor may be costrong with respect to a given $\M$-action in more than one way, or it may not be costrong at all. It is thus desirable to find conditions guaranteeing the existence or uniqueness of costrength.
Regarding existence of costrength, the next Section~\ref{sec:costrong=copointed} provides some conditions in the particular case of a cartesian monoidal category acting upon itself. 
Uniqueness has first been considered by Moggi for strong monads on cartesian categories~\cite{Moggi1991}, and recently generalised by Sato for strong monads on monoidal categories, seen as actegories with the regular action~\cite{Sato2018}. 
Just for completeness, we state below the corresponding dual result for costrengths and endofunctors (not necessarily (co)monads). We should mention, however, that we found that in practice the  hypotheses are rarely verified (see the remark following this proposition): 

\medskip

\begin{proposition}[Uniqueness of costrength~\cite{Moggi1991,Sato2018}]
Let $(\M,\otimes, I)$ be a monoidal category. Assume the following:
\begin{enumerate}
\item The monoidal unit is a cogenerator,%
\footnote{Recall that an object $I$ in a category $\M$ is called a {\em cogenerator} if for any pair of arrows $f,g:X \to Y$ in $\M$, $f=g$ if and only if $h \circ f = h \circ g$ for any $h:Y \to I$. For example, in $\set$, any set with at least two elements is a cogenerator.  
%A category possessing a generator (the dual notion) is also said to be {\em well-pointed}, or {\em having enough points}. 
} 
\item \label{cond-cosep} For any morphism $h:X \otimes Y \to I$ there are $f:X \to I$ and $g:Y \to I$ such that the composite $X \otimes Y \overset{f \otimes g}{\to}I\otimes I \cong I$ coincides with $h$. 
\end{enumerate}
Let $F:\M \to \M$ be a costrong functor with respect to the regular action of $\M$. Then its costrength is uniquely determined by 
\[
\xymatrix{
F(X \otimes Y)
\ar[r]^{\cst}
\ar[d]_{F(f \otimes \id)} 
& 
X \otimes FY 
\ar[r]^-{f \otimes g}
&
I \otimes I 
\ar[d]^{\cong}
\\
F(I \otimes Y)
\ar[r]^{\cong}
&
F(Y)
\ar[r]^{g}
&
I
}
\]
for any objects $X,Y$ and any morphisms $f:X \to I$ and $g:FY \to I$, where the unlabelled isomorphisms refer to the left associators of the monoidal category. 
\end{proposition}

\begin{remark}
The second condition above is guaranteed in the case of (semi-)cocartesian monoidal categories.\footnote{\label{footnote:semicocartesian}A semi-cocartesian category is a monoidal category for which the monoidal unit is also the initial object $\mathbb 0$. Consequently, there are coprojections $X \rightarrow X \otimes Y \leftarrow Y$. If codiagonals $X \otimes X \to X$ exist such that the composites $X \cong X \otimes \mathbb 0 \to X \otimes X \to X$, $X \cong \mathbb 0 \otimes X \to X \otimes X \to X$ are identities, then the monoidal product is also the coproduct and the category is cocartesian. See e.g.~\cite{semicartesian,Fritz2020,semicartesian2} for the dual case of semicartesian categories.} % 
Adding the requirement that the monoidal unit (that is, the initial object) has to be a cogenerator is more restrictive. 
An example comes from propositional logic: in the category of Boolean algebras, $\two$ is both an initial object and a cogenerator. 
Another example arises from the semantics of quantum computing, more precisely from the category of (finite-dimensional) Hilbert spaces. 
Duality is also a source of examples: just take the opposite of any monoidal category for which the monoidal unit is a {\em generator}, such as $(\set, \times, \one)$. 
\end{remark}

%======================================%

%======================================%

\subsection{Costrong functors on cartesian categories}\label{sec:costrong=copointed}

Among monoidal categories occurring in programming, the case that arises most frequently is the one where the monoidal structure is given by the cartesian product, with $(\set,\times)$ as main example. 
It is well known that every $\set$-endofunctor is $(\set,\times)$-strong, in a unique way. 
This naturally raises the question: what can be said about costrong ones?
We shall see that there is an isomorphism between the categories of costrong endofunctors and copointed ones on cartesian categories.

\medskip

Let thus $\M$ be a cartesian category. Consider $\M$ as an $\M$-actegory with the regular action. 
The preliminary result below shows that application of costrength followed by second projection coincides with functor application of this second projection.

\begin{lemma}\label{lem:costrong=copointed}
Let $F:\M \to \M$ be a costrong functor, with costrength $\cst: F(M \times X) \to M \times FX$. 
Then the costrength commutes with the second projection: $\pi_2 \circ \cst = F\pi_2:F(M \times X) \to FX$. 
    \begin{equation}\label{eq:cst-proj2}
    \xymatrix@C=35pt@R=17pt{
    F(M \times X) 
    \ar[r]^-{\cst}
    \ar`d[dr]_(.75){F \pi_2}[dr]
    &
    M \times FX 
    \ar[d]^-{\pi_2}
    \\
    &
    FX
    }
    \end{equation}
\end{lemma}

\begin{proof}
    This follows from
    \[
    \xymatrix{
    &
    F(M\times X) 
    \ar[r]^-{\cst}
    \ar[d]|-{F(! \times \id)}
    %\ar@/_6.75ex/[dd]_-{F\pi_2}
    & 
    M \times FX
    \ar[r]^-{\pi_2}
    \ar[d]|-{! \times \id}
    &
    FX
    \ar@{=}[d]
    \\
    &
    F(\one \times X)
    %\ar@{}[drr]|{(*)}
    \ar[r]^-{\cst}
    \ar[d]|-{F\pi_2}
    &
    \one \times FX
    \ar[r]^-{\pi_2}
    &
    FX
    \ar@{-}@<+.25ex>`d[dll][dll]
    \ar@{-}@<-.25ex>`d[dll][dll]
    %\ar@{<-}`d[ll]`[ll]^-{F\pi_2}
    %\ar@{<-}`d[ll]`u[ull]`[ull]
    %\ar@{<-}`r[rru]_{\Lan_JA\circ B}[uurr]&
    %\ar@{<-}`d`[lll]`[u][lllu]
    %\ar`d`[llllll]`[uu][llllluu]
    \\
    &
    FX 
    \ar@{<-}`l[luu]`[uu]^-{F\pi_2}[uu]
    }
    \]
    where $!$ denotes the unique arrow to the terminal object $\one$ of $\M$. The two top rectangles in the diagram above commute by naturality, while the bottom one is commutative due to the interaction between the costrength and the unit of the monoidal product (the unit being the terminal object). 
\end{proof}

Recall now that an endofunctor $F:\M \to \M$ is called {\em copointed} if it is endowed with a natural transformation $\epsilon:F \to \id_\M$. 
A natural transformation between copointed endofunctors $\alpha:F \to G$ is itself called {\em copointed} if 
    \[
    \xymatrix{
    F \ar[rr]^{\alpha}
    \ar[dr]_{\epsilon}
    &
    &
    G
    \ar[dl]^{\epsilon}
    \\
    &
    \id_{\M}
    &
    }
    \]
    We obtain thus the category of copointed endofunctors and copointed natural transformations, that is, the slice category $[\M,\M]\downarrow \id_\M$. The forgetful functor $[\M,\M]\downarrow \id_\M \to [\M,\M]$ is comonadic if $\M$ is cartesian, with right adjoint sending an arbitrary functor $F:\M \to \M$ to $\id_\M \times F$.

\medskip

Consider also the category $[\M,\M]_{\cst}$ of costrong endofunctors on $\M$ and costrong natural transformations.

\begin{theorem}\label{thm:costrong=copointed}
Let $\M$ be a cartesian category. 
There is an isomorphism between the category $[\M,\M]_{\cst}$ of costrong endofunctors on $\M$, and the category of copointed endofunctors $[\M,\M]\downarrow \id_\M$. 
\end{theorem}

\begin{proof}
    The correspondences are as follows: First, the functor  $\Phi:\mact_\cst(\M,\M) \to [\M,\M]\downarrow \id_\M$ maps $(F, \cst)$ to $(F,\widetilde\epsilon)$ where
    \[
    \widetilde\epsilon: \xymatrix{
    FM 
    \ar[r]^-{\cong}
    &
    F(M \times \one) 
    \ar[r]^{\cst}
    &
    M \times F\one 
    \ar[r]^-{\pi_1}
    &
    M
    }
    \]
    The naturality of $\widetilde\epsilon$ is clear. 
    The functor $\Phi$ acts as identity on arrows. 
    Observe also that the diagram below commutes:
    \[
    \xymatrix@C=35pt{
    &
    F(M \times X) 
    \ar[r]^{\cst}
    \ar`l[dl]_-{F\pi_1}[dl]
    \ar[dd]^{F(\id \times !)}
    &
    M \times FX
    \ar[dd]^{\id \times F!}
    \ar`r[ddr][ddr]^{\pi_1}
    \\
    F M
    \ar`d[dr][dr]^-\cong
    %\ar`l[dd]`[dd]_{\widetilde \epsilon}`[rr][drr]
    &
    &
    &
    \\
    &
    F(M \times \one)
    \ar[r]^{\cst}
    &
    M \times F \one
    \ar[r]^-{\pi_1}
    &
    M
    }
    \]
    Consequently, the first projection links the costrength and the copoint {\em that it induces}:
    \begin{equation}\label{cstcptfstproj1}
    \xymatrix@C=35pt{
    F(M \times X) 
    \ar[d]_{F\pi_1}
    \ar[r]^{\cst}
    & 
    M \times FX
    \ar[d]^{\pi_1}
    \\
    FM
    \ar[r]^{\widetilde\epsilon}
    & 
    M
    }
    \end{equation}
    Using~\eqref{cstcptfstproj1} and the naturality of $\widetilde\epsilon$, we obtain that $\pi_1 \circ \cst = \pi_1 \circ \widetilde \epsilon$, as in the diagram below:    
    \begin{equation}\label{cstcptfstproj2}
    \xymatrix{
    F(M \times X)
    \ar[d]_{\widetilde\epsilon}
    \ar[r]^-{\cst}
    &
    M \times FX 
    \ar[d]^{\pi_1}
    \\
    M \times X 
    \ar[r]^-{\pi_1}
    &
    M
    }
    \end{equation}
    In the opposite direction, the functor $\Psi:[\M,\M]\downarrow \id_\M \to \mact_\cst(\M,\M) $ maps a copointed $(F, \epsilon)$ to $(F, \widetilde\cst)$, with costrength given by
    \[
    \widetilde\cst:\xymatrix@C=50pt{
    F(M \times X)
    \ar[r]^-{\langle F\pi_1,F\pi_2\rangle}
    &
    F M \times F X 
    \ar[r]^-{\epsilon \times \id}
    &
    M \times F X
    }
    \]
    where $\langle-,-\rangle$ denotes the pairing into the product.
    That $\widetilde\cst$ is natural and satisfies~\eqref{eq:colax-morphism} follows from the fact that any functor on a cartesian category is automatically colax monoidal, and so is any natural transformation.
    Also $\Psi$ acts as identity on arrows. 
    The commutativity of 
    \[
    \xymatrix@C=50pt{
    F(M\times X)
    \ar[r]^-{\langle F\pi_1, F\pi_2\rangle}
    \ar`d[dr]_-{F\pi_1}[dr]
    &
    F M \times F X
    \ar[r]^{\epsilon \times \id}
    \ar[d]^{\pi_1}
    &
    M \times F X
    \ar[d]^{\pi_1}
    \\
    &
    F M 
    \ar[r]^{\epsilon}
    &
    M}
    \]
    shows that Equations~\eqref{cstcptfstproj1} and~\eqref{cstcptfstproj2} hold also for a copointed functor $(F,\epsilon)$ and the {\em induced} costrength $\widetilde\cst$. In particular, $\pi_1 \circ \widetilde \cst = \pi_1 \circ \epsilon$. 
    This in turn entails $\Psi \circ \Phi=\id$: starting with $(F,\cst)$, constructing $\Phi(F,\cst)=(F,\widetilde \epsilon)$ and subsequently $\Psi(F,\widetilde \epsilon)= (F, \widetilde{\widetilde{\cst}})$, we have that 
    \[
    \pi_1 \circ \widetilde{\widetilde{\cst}} = \pi_1 \circ \widetilde\epsilon = \pi_1 \circ \cst
    \]
    by Equation~\eqref{cstcptfstproj2}, and $\pi_2 \circ \widetilde{\widetilde{\cst}} =F\pi_2 = \pi_2 \circ \cst$ by Lemma~\ref{lem:costrong=copointed}. Hence $\widetilde{\widetilde{\cst}}=\cst$.

    Finally, the relation $\Phi \circ \Psi = \id$ follows from 
    \[
    \xymatrix@C=50pt{
    F M 
    \ar[r]^-{\cong} 
    \ar[d]_{\epsilon}
    & 
    F(M\times \one) 
    \ar[r]^-{\langle F\pi_1, F\pi_2\rangle}
    \ar[d]^{\epsilon}
    &
    F M\times F \one
    \ar[r]^{\epsilon \times \id}
    \ar[d]^{\epsilon \times \epsilon}
    &
    M \times F \one
    \ar[d]^{\pi_1} 
    \ar[dl]_-{\id \times \epsilon}
    \\
    M
    \ar[r]^{\cong}
    &
    M \times \one 
    \ar@{=}[r]
    &
    M \times \one 
    \ar[r]^{\pi_1}
    &
    M
    &
    }
    \]
    where we used that any natural transformation is automatically colax monoidal with respect to the cartesian product.
    
\end{proof}

The gist of the above theorem is that on a cartesian category, to give a costrength for an endofunctor is the same as to give a {\em copoint}, that is, a natural transformation with codomain the identity functor. 
In particular, we see now that costrengths for the {\tt Reader} monad $[S,-]$ on $(\set,\times, \one)$ (or any other cartesian category) are in one-to-one correspondence with natural transformations $[S,X]\to X$, which by Yoneda lemma, correspond to points (elements) of $S$. 
Regarding the {\tt Writer} functor (comonad) $S \times -$, another application of the Yoneda lemma shows that there is only one natural transformation $S \times X \to X$, namely the second projection, hence there is only one way in which the {\tt Writer} functor can be $(\set,\times, \one)$-costrong.

\begin{corollary} 
Comonads on cartesian categories are costrong.
\end{corollary}

\begin{corollary}\label{cor:cofree-cst}
For every endofunctor $F$ on a cartesian category $\M$, the {\em cofree} copointed functor over $F$, namely $\id_\M \times F$~\cite{LPW2000}, is costrong. 
\end{corollary}

The above two corollaries show that there is a plethora of examples of costrong endofunctors on cartesian categories. In particular, the counit of a comonad produces a costrength, but a comonad may be costrong in a different way, see the case of the $\mathtt{Costate}$ comonad from Example~\ref{ex:cst-functors}.(\ref{ex:costrong-comonad-not-unique-cst}).

\begin{remark}
The one-to-one correspondence between $(\set,\times,\one)$-costrong endofunctors and copointed ones has the following intuitive explanation: 
In $\set$, or more generally in any cartesian category, having a functor $F$ endowed with a costrength $F(M\times X) \to M \times F X$ says that from an $F$-computation depending on both $M$ and $X$, the $M$-part can be extracted untouched, while retaining an $F$-computation on the $X$-part. 
In a general monoidal category, there is no canonical way to get $M$ out of $M \otimes X$, so a costrength is substantial extra {\em structure}.
However, in a cartesian category there is a canonical projection $\pi_1: M \times X \to M$. 
Therefore, once $F$ comes equipped with a copoint $\epsilon_{M \times X} : F(M\times X)\to M\times X$, we can automatically produce the $M$-component by composing with the projection. 
The $F X$-component is again obtained by projection, but this time {\em inside} the $F$-computation $F\pi_2:F(M\times X)\to F X$.
In conclusion, in the cartesian case, a costrength carries no more information than the canonical map $\epsilon_X:F X \to X$ -- that is, the ability to extract underlying values from any $F$-computation. 
The full costrength is then reconstructed from $\epsilon$ and the projections. 
However, that does not make copointed functors meaningless. In fact, these correspond to effects that are observational, context-like, or read-only.
\end{remark}

\begin{remark}
One of the reviewers asked whether the results of this section extend to costrong functors between {\em distinct} cartesian categories. Albeit this departs from our object of study (namely costrong {\em endo}functors), we shall however make a few comments on the subject.

First, we need an action of the {\em same monoidal category} on both cartesian categories. 
One way to achieve this is the following scenario: consider two cartesian categories $(\A,\times,\one),(\B,\times,\one)$ and a lax monoidal functor $J: \A\to \B$. 
Then $(\A,\times)$ acts upon itself by the regular action, and also on $\B$ using the functor $J$, by $M \cdot Y  = J(M)\times Y$. 
By repeating the arguments of this section, one can see that $(\A,\times)$-costrengths $F(M \times X) \to JM\times FX$ on a functor $F:\A\to \B$ are in one-to-one correspondence with natural transformations $F\to J$. 
We may call such functors {\em $J$-relative costrong}. In particular, $J$-relative comonads~\cite{AltenkirchChapmanUustalu-fossacs2010} are {\em $J$-relative costrong}.
Let now $J$ be the Yoneda embedding $\M \to [\M\op,\set]$ with $\M$ cartesian. 
It is well-known that $J$ preserves all limits (that exist in $\M$), in particular products. 
Then a functor $F:\M \to [\M\op,\set]$ is (the curried version of) a profunctor $\M \nto \M$. 
If $\M$ has arbitrary copowers, more generally arbitrary coproducts, then a $J$-relative costrength for $F$ is equivalent to an {\em evaluation}-like natural transformation from the $F(M)(N)$-copower of $N$ to $M$.
Examples include procomonads, such as the one induced by the adjunction $G_*\dashv G^*$, for any functor $G:\M\to \M$ (see the paragraph on profunctors in the next section). 

More generally, our arguments carry mutatis mutandis for a span of cartesian categories and lax monoidal functors $J:(\M,\times)\to (\A,\times)$, $K:(\M,\times)\to (\B,\times)$. The functors $J$ and $K$ determine (lax) $(\M,\times)$-actions on $\A$ and $\B$ as above, and a functor $F:\A \to \B$ is costrong with respect to these structures if and only if it comes equipped with a natural transformation $FJ\to K$:
\[
\xymatrix@C=60pt@R=15pt{
& 
\A
\ar[dd]^F
\\
\M 
\ar[ur]^-J
\ar[dr]_-K
\ar@{}[r]|{\Downarrow}
&
\\
&
\B 
}
\]

\end{remark}

\begin{example}
Consider a polynomial functor $F$ on $\set$, described by a set $S$ (of {\em shapes}, see~\cite{NiuSpivak}), and for each $s\in S$, by a family $\mathsf{ar}(s)$ (of {\em positions}). Explicitly, $FX = \coprod_{s\in S} X^{\mathsf{ar}(s)}$. 
Then a copoint of $F$, that is, a natural transformation $F \Rightarrow \id_\set$, is, by the Yoneda lemma, completely determined by specifying, for each $s\in S$, an element of $\mathsf{ar}(s)$. In particular, polynomial functors containing constants (i.e., $s\in S$ with $\mathsf{ar}(s)=\emptyset$) have no copoints, hence are not costrong. This is why the  $\mathtt{Maybe}$ or the $\mathtt{List}$ monads from functional programming cannot be $(\set,\times,\one)$-costrong. 
\end{example}

\begin{remark}
A careful inspection of the diagrams involved in the proofs of Lemma~\ref{lem:costrong=copointed} and of Theorem~\ref{thm:costrong=copointed} shows that the hypothesis that $\M$ is a cartesian category can be relaxed to a semicartesian category~\cite{semicartesian,Fritz2020,semicartesian2}. A semicartesian category is a monoidal category where the unit object is terminal. Equivalently, it can be described as a ``monoidal category with projections'' (see also Footnote~\footref{footnote:semicocartesian}). 
The arguments above carry unchanged in the semicartesian case, under the proviso that copointed endofunctors should be replaced by copointed colax monoidal endofunctors (where both the functor component and the ``copoint'' component are colax monoidal). 

\end{remark}

The equivalence between costrong functors and copointed ones in the cartesian case allows us to see that the category of costrong functors $\mact_\cst(\M,\M)$ supports several monoidal structures, as follows: 
First, it is strictly monoidal with functor composition. 
Next, by the results of the subsequent Section~\ref{sec:acts = algs}, $\mact_\cst(\M,\M)$ has coproducts;
Finally, for $\M=\set$, if $F$ and $G$ are costrong functors, hence copointed, then their Day convolution $F * G$ is again copointed, hence costrong. The copoint is obtained as the composite
\[
(F*G)X = \int^{Y,Z\in \set} [Y \times Z,X]\times FY \times GZ \, \to \int^{Y,Z\in \set} [Y \times Z,X]\times Y \times Z \cong X
\]

%======================================%
\subsection{Costrong functors, categorically}\label{sec:acts = algs}

In this section we develop the category-theoretic foundations of costrong functors. 
Readers primarily interested in applications may wish to proceed directly to the final section of the paper, which can be read independently of the details developed here.

%======================================%

\paragraph{Costrong functors as colax algebra morphisms.}
Some of the results below are direct consequences of the fact that $\M$-actegories are the pseudo-algebras for the pseudomonad $\M \times -$ on $\mathbf{Cat}$, while $\M$-(co)strong functors between $\M$-actegories are the corresponding (co)lax morphisms of algebras (and the same goes for (co)strong natural transformations).

\begin{proposition} 
Let $(\M, \otimes,I)$ be a braided (or symmetric) monoidal category. Consider $\M$ as an $\M$-actegory, with the regular action.

\begin{enumerate}
\item There is an adjunction between $\M$ and the category $\mact_\cst(\M,\M)$ of $\M$-costrong endofunctors on $\M$ and $\M$-costrong natural transformations:
\[
\xymatrix@C=60pt{
\M
\ar@<+1.1ex>[r]^-{M \mapsto M \otimes -} 
\ar@{}[r]|-{\top}
&
\mact_\cst(\M,\M)
\ar@<+1.1ex>[l]^-{F \mapsto F(I)}
}
\]

\item (\cite{HeunenKarvonen2015}) This adjunction lifts to an adjunction between the categories of comonoids in $\M$ and $\M$-costrong comonads on $\M$:
\[
\xymatrix@C=60pt{
\mathbf{Comonoids}(\M)
\ar@<+1.1ex>[r]^{C \mapsto C \otimes -} 
\ar@{}[r]|-{\top}
&
\mathbf{Comonads}(\M)_\cst
\ar@<+1.1ex>[l]^{\mathbb D \mapsto \mathbb D(I)}
}
\]
\end{enumerate}
\end{proposition}

\begin{proposition}[Doctrinal adjunction~\cite{Kelly-Doctrinal}]
Let $\M$ be a monoidal category and $\A,\B$ two $\M$-actegories related by an adjunction $F \dashv G:\B \to \A$. Then the following hold:

\begin{enumerate}

\item The left adjoint $F$ is $\M$-costrong if and only if the right adjoint $G$ is $\M$-strong.

\item The adjunction lifts to an adjunction in the $2$-category $\M$-$\act_{\cst}$ if and only if the strength of the right adjoint $G$ is a natural isomorphism. \end{enumerate}
\end{proposition}

\begin{proposition}[Non-canonical isomorphism~\cite{LucatelliNunes2019}]
Let $(F,\cst)$ be a costrong functor between $\M$-actegories $\A$ and $\B$. If there exists an isomorphism $F(M\cdot X) \to M\cdot F X$, natural in both $M$ and $X$, then $\cst:F(M\cdot X) \to M\cdot F(X)$ is also a natural isomorphism. 
\end{proposition}

\paragraph{Colimits of costrong functors.} 
Let $\A$ be a cocomplete $\M$-actegory. Then $[\A,\A]$ is again cocomplete, with colimits computed pointwise~\cite[Sec.~3.3]{Kelly-book}. 

\begin{remark}\label{rem:colim-cst}
It is easy to see that the forgetful functor $\mact_\cst(\A,\A) \to [\A,\A]$ creates colimits. That is, a colimit of costrong functors is again costrong, with costrength induced by the universality of the colimit. 
In particular, coproducts of $\M$-costrong functors are again costrong, and the coproduct injections are costrong natural transformations. 
Once again, category theory provides us with ``Theorems for free'', as Philip Wadler famously put it. 
Compare the above result to \cite[Lemma 3.5]{JaskelioffRypacek2012} where instead a direct proof is provided to show that traversable functors (that is, the costrong functors with respect to the action of applicative functors on $\set$ by evaluation) are closed under arbitrary coproducts. 
\end{remark}

%======================================%

\paragraph{Free monads inherit costrength.}
Given a functor $F:\A \to \A$, recall that a monad $T$ on $\A$ is (algebraically) free on $F$ if the forgetful functor from the category of $F$-algebras $\mathbf{Alg}(F) \to \A$ is right adjoint, and the monad it generates is precisely $T$~\cite{Kelly1980}. %section 22 
In~\cite{McDermottUustalu2022}, it is shown that the strength of a functor does not necessarily pass to the associated algebraically free monad (if it exists). 
On the contrary, algebraically free monads inherit costrength from their generating functors, as we show next.

\begin{proposition}\label{prop:cst-free-mnd}
Let $\A$ be a locally presentable $\M$-actegory and let $F:\A\to \A$ be an accessible $\M$-costrong functor. Then the free monad $T_F$ on $F$ exists, is an $\M$-costrong monad and the universal natural transformation $F\to T_F$ is costrong. 
\end{proposition}

\begin{proof}
The existence of the free monad on $F$ under the given hypotheses is well known. 
In particular, it  can be obtained as $T_{F}X = \mu Y.X+FY$. 
It remains to show that it is an $\M$-costrong monad.
The $\M$-costrength of $T_F$, denoted $\cst^{T_F}$, is obtained as the unique morphism from the initial algebra in the diagram below. 
\[
\xymatrix@C=38pt{
M \cdot X+ FT_F(M \cdot X)
\ar[d]_{\id + \cst^T}
\ar[rrr]^{\mathsf{in}}
& 
&
&
T_F(M \cdot X)
\ar@{.>}[d]^{\cst^T}
\\
M \cdot X+ F(M \cdot T_F X)
\ar[r]^{\id+\cst}
& 
M \cdot X + M \cdot F T_F X 
\ar[r]^-{[\id\cdot \mathsf{inl}, \id\cdot \mathsf{inr}]}
& 
M \cdot (X+F T_F X)
\ar[r]^-{\id\cdot \mathsf{in}}
& 
M \cdot T_F X
}
\]
In the diagram above, the top arrow $\mathsf{in}$ is the initial algebra isomorphism, while $\mathsf{inl}$, $\mathsf{inr}$ are the coproduct injections.  
Standard but tedious verifications show that $\cst^{T_F}$ is natural both in $M$ and in $X$, that $(T_F,\cst^T)$ is indeed an $\M$-costrong monad, and that the natural transformation $F \to T_F$ asserting the freeness of $T_F$ over $F$ is also costrong, being a composite of such. % 
\end{proof}

\paragraph{Costrong functors and liftings to (co)algebras.}

Let again $\A$ be an $\M$-actegory. 
Then $\M$-costrengths on a functor $F:\A \to \A$ are in one-to-one correspondences with liftings of the $\M$-action from $\A$ to the category of $F$-algebras $\mathbf{Alg}(F)$.\footnote{Actually, this holds even in the case of an $\M$-graded monad, not necessarily an action. } 
This is analogue to the lifting of the monoidal structure to the Eilenberg-Moore category of algebras for a comonoidal monad~\cite{Moerdijk2002}, or to the lifting of an $\M$-action to the Eilenberg-Moore category of algebras of an $\M$-monad~\cite{Skoda2004}. 
There is also a dual perspective: for every object $M$ of the monoidal category $\M$ there is an endofunctor $M\cdot -$ on every $\M$-actegory. Then an $\M$-costrength for a functor $F:\A \to \B$ induces a lifting of $F$ between the categories of coalgebras for $M \cdot -$, that is, between monoidal stream $M$-automata. Moreover, this lifting is uniform in $M$. More details will be discussed in Section~\ref{sec:costrong-stream}, specifically regarding the case of the cartesian category of sets acting on itself. 

%%===============================%

\paragraph{(Co)Strong functors induce strong profunctors.} 

Let $\A,\B$ be $\M$-actegories. 
Recall briefly that a profunctor $P:\A \nto \B$ is a functor $P:\B\op \times \A \to \set$. A morphism of profunctors is a natural transformation. Profunctors $A\nto \B$ and their morphisms organise themselves into a category $\mathbf{Prof}(\A,\B)$. 
Any functor $F:\A \to \B$ induces a pair of adjoint profunctors, $F_*:\A \nto \B$, $F^*:\B \nto \A$ with $F_* \dashv F^*$, given by $F_*(Y,X) = \B(Y,F(X))$ and $F^*(X,Y)= \B(F(X),Y)$. 
In particular, the assignment $F \mapsto F^*$ determines a functor $[\A,\B]\to \mathsf{Prof}(\B,\A)\op$ that we shall later need.  
If $\A$ and $\B$ are $\M$-actegories, a profunctor $P:\A \nto \B$ is called $\M$-strong (or a Tambara module) if it is endowed with a family of maps $\st: P(X,Y) \to P(M\cdot X, M \cdot Y)$, natural in $X$ and $Y$ and dinatural in $M$, satisfying the usual coherence laws~\cite{PastroStreet2008,JaskelioffRivas2017}. Strong profunctors have been successfully employed in modelling equilibria and strategies of open games~\cite{AGGKLNF2021}. Hughes' arrows from functional programming are monoids in the monoidal category of strong profunctors (with composition as the tensor product).

\begin{proposition}\label{prop:(c)st-functor-to-st-profunctor}
Let $\A,\B$ be $\M$-actegories and $F:\A\to \B$ be an arbitrary functor. Then $F$ is $\M$-strong if and only if $F_*$ is an $\M$-strong profunctor. Dually, $F$ is an $\M$-costrong profunctor if and only if $F^*$ is an $\M$-strong profunctor.\footnote{Notice that the variance of strengths for profunctors does not change, both $F^*$ and $F_*$ being $\M$-strong.}
\end{proposition}

\begin{proof}
As our interest lies primarily in $\M$-costrong functors, we shall present the argument for these, the dual case being similar.
We shall use the following observation, easily obtained by the Yoneda lemma:
given functors $S:\A \to \A$, $F:\A\to\B$, $T:\B\to \B$, natural transformations $FS \to TF$ as in the diagram below
\[
\xymatrix{\A \ar[r]^-S \ar[d]_F \ar@{}[dr]|{\Leftarrow}
  &
  \A \ar[d]^F
  \\
  \B \ar[r]^-T 
  &
  \B 
}
\]
are in bijection with natural transformations between profunctors $\B(F-,=) \to \B(FS-,T=)$.

Hence to give a natural transformation $F(M \cdot - ) \to M\cdot F(-)$ is the same as to give $F^*(-,=)\to  F^*(M\cdot(-),M\cdot(=))$.
For reader's sake, we spell out below the correspondences:
If $F$ is $\M$-costrong, the strength for the right adjoint profunctor $F^*$ is given by
\[
\xymatrix{F^*(X,Y) = \B(FX,Y) \ar[r] & \B(M \cdot FX, M \cdot Y) \ar[r]^-{\B(\cst,\id)} & \B(F(M\cdot X),M \cdot Y)=F^*(M \cdot X,M\cdot Y)
}
\]
where the unlabelled arrow represents the action of the functor $M \cdot -$ on arrows. Conversely, the $\M$-strength of the profunctor $F^*$ determines a costrength on $F$ as the image of the identity via the composite
\[
\xymatrix{\B(FX,FX) = F^* (X,FX)\ar[r]^-\st & F^* (M \cdot X, M \cdot FX) = \B(F(M \cdot X),M \cdot FX)}
\]
Straightforward verifications show now that these correspondences indeed produce the desired (co)strengths.
\end{proof}

Our interest in profunctors comes from the fact that universal %(co)free 
$\M$-strong profunctors are easy to find. 
The following result goes back to~\cite{PastroStreet2008}, but see also~\cite{JaskelioffRivas2017}:

\begin{proposition}
The forgetful functor from the category of $\M$-strong profunctors and strong natural transformations 
\[
\xymatrix@!{
\mathbf{Prof}(\A,\A)_\st \ar[r] & \mathbf{Prof}(\A,\A)}
\]
has both a left and a right adjoint, given by 
\[
\mathsf{Free}_{\st}(P)(Y,X) = \int^{M,X',Y'} \A(Y,M\cdot Y') \times \A(M \cdot X',X) \times P(Y,X)
\]
respectively 
\[
\mathsf{Cofree}_{\st}(P)(Y,X) = \int_M P(M\cdot Y,M\cdot X)
\]  
whenever the coend and the end above exist (which is the case, for example, whenever $\M$ is small and $\A$ is (co)complete). 
\end{proposition}

In view of the result above and of Proposition~\ref{prop:(c)st-functor-to-st-profunctor}, it is natural to ask about universal constructions of costrong functors. 
More precisely, whether the forgetful functor from the category of costrong endofunctors $\mact_\cst(\A,\A)$ over an $\M$-actegory $\A$ to the category of all endofunctors $[\A,\A]$ has a left or right adjoint. In view of Remark~\ref{rem:colim-cst} and Proposition~\ref{prop:cst-free-mnd}, we would rather expect a right adjoint (that is, a method of producing a {\em cofree} costrong functor from a given functor). 

Consider now the diagram below, commutative with respect to the horizontal forgetful functors: 
\[
\xymatrix@C=40pt{
[\A,\A]_{\cst} 
\ar[d]_{(-)^*} 
\ar[r] 
& 
[\A,\A] 
\ar[d]^{(-)^*} 
\\ 
\mathbf{Prof}(\A,\A)\op_\st 
\ar[r]^\perp_\perp 
\ar@{<-}@/^2.5ex/[r]^-{\mathsf{Free}_{\st}}
\ar@{<-}@/_2.5ex/[r]_-{\mathsf{Cofree}_{\st}}
& 
\mathbf{Prof}(\A,\A)\op  
}
\]
The bottom forgetful functor is both left and right adjoint by the previous proposition. It is still under investigation under which conditions at least one of the bottom adjunctions lift to the top; that is, when {\em (co)free} costrong functors exist and how to construct them.%
\footnote{Compare with~\cite{JaskelioffRivas2017}, where {\em free} applicative functors (endofunctors on $\set$, with compatible strong and lax monoidal structure) were obtained. In fact, freeness in {\em op. cit.} referred to the lax monoidal structure of the respective functor, not to the strength which {\em always} exists.} %
The well-known results on lifting adjoints~\cite{Dubuc1968,Johnstone1975} do not seem to apply directly here.

%======================================%

\section{Applications}\label{sec:applications}

%======================================%

\subsection{Costrong functors and (mixed) optics}\label{sec:optics}

Optics are bidirectional data accessors, offering a powerful abstraction for accessing data in a compositional manner. 
For a programmer's perspective, one can intuitively think of an optic as decomposing a type or data structure through a monoidal action, enabling access to a nested structure. 
The result of the decomposition is a context or residual, which can later recompose with other data to get a composite object. 
We are considering here mixed optics for a pair of actions of a monoidal category $\M$; these work independently from one another and only interact through the residuals. 

\medskip

While the existing literature focused on changing optics via a single monoidal functor preserving the tensor product up to an isomorphism and acting on all components, here we allow pairs of $\M$-morphisms of mixed strength, acting individually on each $\M$-action of the optic rather than simultaneously on both. 

\medskip

A careful inspection in the definition of optics (to be recalled below) will convince the reader that the two actions can in fact be relaxed to a graded monad and a graded comonad, connected via the monoidal category $\mathcal{M}$ that provides the grading.\footnote{More generally, via a mixed {\em graded} distributive law~\cite{faul}.}  
We leave the exploration of optics arising from a graded monad-comonad pair for future research, and focus now on genuine $\M$-actions. 

\medskip

Let $\M$ be a monoidal category acting on two categories $\A$ and $\B$.
The category of (mixed) optics $\optic_{\A,\B}$\footnote{Strictly speaking, optics depend not on the {\em underlying categories} $\A$, $\B$ but on their {\em $\M$-actions}.} has as objects pairs of objects of $\A$ and $\B$, respectively, and hom-sets
\[
\optic_{\A,\B}((X',Y'),(X,Y)) = \int^{M\in \M} \A(X',M\cdot X)\times \B(M\cdot Y,Y')
\]

We refer to~\cite{CEGLMPR2020} for more details and a plethora of examples; these include cases where $\A,\B$ are not necessarily small, but the above coend still exists. There is an identity-on-objects functor $\A\op \otimes \B \to \optic_{\A,\B} $ exhibiting $\optic_{\A,\B}$ as the Kleisli category of a monad in the bicategory of profunctors~\cite{PastroStreet2008}.

\medskip

Observe that optics are in fact {\em parametric} in the two $\M$-actions, in the sense that pairs of $\M$-costrong/$\M$-strong functors act as {\em optics transformers}. 
More in detail, an $\M$-costrong functor between $\M$-actegories $F:\A \to \A'$ and an $\M$-strong one $G:\B \to \B'$ induce  
a functor between categories of optics $\optic(F,G):\optic_{\A,\B}\to \optic_{\A', \B'}$ as follows: on objects, $F$ and $G$ act component-wise by $(X,Y) \mapsto (FX,GY)$. The action on morphisms, that is, on optics, is globally induced by the composite 
\[
\xymatrix@C=8pt{
{\A}(X',{M}{\cdot} X){\otimes} {\B}({M}{\cdot} Y,Y') 
\ar[r]
&
{\A}'(F X',F({M}{\cdot} X)){\otimes} {\B}'(G({M}{\cdot} Y),G Y') 
\ar[r]
&
\A'(F X',{M}{\cdot} F X){\otimes} \B'({M}{\cdot} G Y,G Y') 
}
\]
The correspondence $(\A,\B)\mapsto \optic_{\A,\B}$ determines now a $2$-functor 
    \[
    \optic: \mact_\cst \times \mact_\st \to \cat
    \]
    from the product of the $2$-categories of $\M$-actions, $\M$-costrong functors and $\M$-costrong natural transformations $\mact_\cst$, respectively $\M$-actions, strong functors and strong natural transformations $\mact_\st$ into $\cat$.

%======================================%

\subsection{Costrong functors and streams}\label{sec:costrong-stream}%

In the sequel, we shall consider $(\set,\times, \one)$ with the cartesian closed structure acting upon itself.
Let $F:\set \to \set$ a $(\set,\times, \one)$-costrong functor.
As explained in Section~\ref{sec:acts = algs}, we can interpret the costrength of $F$ as a distributive law of $F$ over the {\tt Writer} functor $T_M X = M \times X$, {\em uniform} in the parameter $M$. That is, a stream GSOS specifications~\cite{Klin2011}.
This distributive law allows $F$ to lift to $T_M$-coalgebras. 
Recall that a $T_M$-coalgebra is a set $X$, endowed with a map $\langle {\mathtt{out}}, {\mathtt{next}}\rangle: X \to M\times X$, that is, a stream automata. 
The final $T_M$-coalgebra is the set of streams $M^\omega$, with coalgebra structure given by the usual pairing of maps $\langle {\mathtt{head}}, {\mathtt{tail}}\rangle: M^\omega \to M\times M^\omega$. 
The unique map from the $T_M$-coalgebra $X$ into $M^\omega$ sends a state $x\in X$ to the stream of outputs originated in that state: 
\[
\llbracket x
\rrbracket  = (\mathtt {out}(x), \mathtt{out}(\mathtt {next}(x), \mathtt{out}(\mathtt {next}(\mathtt{next}(x)), \ldots )
\] 
Let us now examine the effect of applying an arbitrary functor $F$ to the coalgebra. 
Intuitively, one can think of $F$ as introducing a context that encapsulates the transition system. 
While the internal dynamics persist, it becomes hidden from an external observer due to this contextual wrapping. 
Generally, there is no way of extracting the outputs from within the context, unless $F$ is $\M$-costrong:  
\[
\xymatrix@C=50pt{
F X 
\ar[r]^-{F\langle {\mathtt{out}}, {\mathtt{next}}\rangle}
&
F(M \times X)
\ar[r]^{\cst}
&
M \times F X
}
\]
Hence now the stream of $M$-outputs is preserved exactly as it was originally, remaining accessible, although the transition to the next state happens inside the context. 
This behaviour is formally captured by the final coalgebra, as follows:
\[
\xymatrix@C=50pt{
X 
\ar[r]^-{%
\langle {\mathtt{out}}, {\mathtt{next}}\rangle
} 
\ar[d]_-{\llbracket -
%{\mathtt{out}}, {\mathtt{next}} 
\rrbracket }
& 
M \times X
\ar[d]^{\id \times \llbracket - 
%{\mathtt{out}}, {\mathtt{next}} 
\rrbracket }
&
\phantom{aa}F X \phantom{i}
\ar@{}[l]^(.25){}="a"
\ar[r]^-{F\langle {\mathtt{out}}, {\mathtt{next}}\rangle}
\ar[d]_-{F\llbracket -
%\langle {\mathtt{out}}, {\mathtt{next}}\rangle 
\rrbracket}
&
F(M \times X)
\ar[r]^{\cst}
\ar[d]
&
M \times F X
\ar[d]^-{\id \times \llbracket -%\cst \, F\langle {\mathtt{out}}, {\mathtt{next}}\rangle 
\rrbracket}
\\
M^\omega 
\ar[r]^-{%
\langle {\mathtt{head}}, {\mathtt{tail}}\rangle
} 
& 
M \times M^\omega
&
F M^\omega 
\ar[r]^-{%
F\langle {\mathtt{head}}, {\mathtt{tail}}\rangle
} 
\ar[d]_-{\llbracket -
%\cst \, F\langle \mathtt{head},\mathtt{tail}\rangle 
\rrbracket}
&
F(M \times M^\omega)
\ar[r]^\cst
&
M \times F M^\omega
\ar[d]
\\
&&
M^\omega 
\ar@{}[l]^(.25){}="b" 
\ar@<+0.75ex>@{<-}`l[uu]`[uu][uu] %"b";"a"
%\ar@{<-}`l[uu]`[uu][uu]
\ar[rr]^-{%
\langle {\mathtt{head}}, {\mathtt{tail}}\rangle
} 
& &
M \times M^\omega
}
\]

Another well-known coalgebraic principle captured using  $(\set, \times, \one)$-costrong functors is {\em coinduction up-to}: 
First, observe that the costrength of $F$ induces in particular an $F$-algebra structure on the final $T_M$-coalgebra of streams, denoted above by $\llbracket-\rrbracket: F(M^\omega) \to M^\omega$. 
Next, given a function $f: X \to M^\omega$, it is often the case that there is no $T_M$-coalgebra structure on $X$ such that the associated behaviour map to be precisely $f$. 
Coinduction up-to remedies in many cases this issue: the  $F$-algebra structure on streams guarantees, for any $M \times F(-)$-coalgebra $\phi:X \to M \times F X$, a unique arrow $\tilde \phi: X \to M^\omega$ making the diagram below commute~\cite{Bartels2003}:  
\[
\xymatrix@R=18pt@C=50pt{
&
& 
F M^\omega
\ar[d]^{\llbracket - \rrbracket}
\\
X 
\ar@{.>}[rr]^{\exists ! \tilde \phi}
\ar[d]_{\phi}
&
&
M^\omega 
\ar[d]^{\langle {\mathtt{head}}, {\mathtt{tail}}\rangle}
\\
M \times F X
\ar[r]^{\id \times F(\tilde \phi)}
&
M \times F M^\omega
\ar[r]^{\id \times \llbracket - \rrbracket}
&
M \times M^\omega
}
\]
If the original $f:X \to M^\omega$ happens to satisfy the diagram, then by uniqueness, there is a coinductive definition for $f$ (up-to $\llbracket-\rrbracket$).

%======================================%

\section{Conclusions}

The authors' interest in distributive laws over graded monads/costrong functors with respect to actions of monoidal categories emerged when studying transformations of profunctor optics as pairs of strong/costrong functors between actegories. 
By dualising the well-established notion of strong functors, we observed that costrength naturally arises in a variety of categorical constructions, particularly when considering coalgebraic structures and context interaction. 
Our results point that costrength can be fruitfully understood through the lens of graded distributive laws, and that it provides a principled way to mediate between monoidal actions and coalgebraic behaviour specifications, beyond abstract formalism.

\medskip

The reader may ask why we considered only graded distributive laws which are in one-to-one correspondence with Kleisli liftings, and not their Eilenberg--Moore counterparts.
After all, there is also a well-established topic of (ungraded) distributive laws versus Eilenberg-Moore liftings, which generalises without difficulty to the graded setting.
One reason is that we did not want to divert reader's attention from costrong functors. 
Another reason is that when the graded monad corresponds to the regular of a monoidal category upon itself, graded Eilenberg-Moore distributive laws produce strong functors, which are already well-covered in the literature.
However, in the general setting of graded monads or actegories, the study of graded Eilenberg–Moore distributive laws remains an interesting topic that we plan to pursue in future work.
 
\medskip

There are many other directions that remain open for further exploration, like the development of a comprehensive theory of costrong functors, the study of costrength/graded distributive laws in enriched and higher-categorical contexts, potential applications in effect-coeffect systems and compositional semantics, not to mention the connections with various systems of graded modal or graded linear logic.

\medskip

As explained in Remark~\ref{rem:kl-bicat}.\ref{rem:kl-bicat2}, the Kleisli construction of a graded monad introduced in~\cite{FujiiKatsumataMelies-FoSSaCS2016} is inherently bicategorical in nature, and it is this the proper context in which the results of this paper should be seen. 
On a different perspective, graded monads generalize monads, and so should do their Kleisli construction. 
On one hand, in~\cite{FujiiKatsumataMelies-FoSSaCS2016}, this is exhibited as a genuine Kleisli construction for a genuine monad {\em in} a $2$-category. 
On another hand, the Kleisli category of a monad is the {\em lax colimit} of the monad seen as a lax functor. 
As graded monads generalise monads, it is expected the same to happen in the graded case. 
However, the lax colimit, much resembling Grothendieck's construction, is not the one described in~\cite{FujiiKatsumataMelies-FoSSaCS2016}. 
It is the {\sf coPara} construction of~\cite{CGHR-copara}. 
While in the latter the grades only decorate morphisms (as {\em co}parameters), the Fujii-Katsumata-Melli\`es construction consider grades on both object-level, as well as morphism-level. 
The precise connection between these two Kleisli constructions is still under investigation.

%======================================%
\section*{Acknowledgments} We thank the anonymous reviewers for their valuable suggestions and insightful comments.

%======================================%

\bibliographystyle{alpha}

%======================================%

\end{document}